\documentclass{amsart}
\allowdisplaybreaks
\usepackage{slashed}
\usepackage{mathrsfs}
\usepackage{amsmath,amsthm}
\usepackage{amssymb}
\usepackage{MnSymbol, wasysym}
\usepackage[margin=1.4in,dvips]{geometry}
\usepackage{hyperref}
\usepackage[all]{xy}
\usepackage{tikz}
\def\tr{\mbox{tr}}

\def\f12{\frac 1 2}

\def\f12{\frac 1 2}

\DeclareMathOperator{\proj}{proj}

\makeindex

\newtheorem{definition}{Definition}[section]

\newtheorem{remark}{Remark}[section]
\newtheorem{lemma}{Lemma}[subsection]
\newtheorem{theorem}{Theorem}[section]
\newtheorem*{theorem*}{Theorem}
\newtheorem*{corollary*}{Corollary}
\newtheorem{proposition}{Proposition}[subsection]

\newtheorem{corollary}{Corollary}[section]

\title{Stationary axisymmetric black holes with matter}
\author{Otis Chodosh}
\address{University of Cambridge, DPMMS, Wilberforce~Road,~Cambridge,~UK,~CB3~0WB}
\author{Yakov Shlapentokh-Rothman}
\address{Princeton University, Department of Mathematics, Fine~Hall,~Washington~Road,~Princeton,~NJ~08544}
\date{\today}
\begin{document}
\maketitle
\begin{abstract}
We provide a geometric framework for the construction of non-vacuum black holes whose metrics are stationary and axisymmetric. Under suitable assumptions we show that the Einstein equations reduce to an Einstein-harmonic map type system and analyze the compatibility of the resulting equations.

This framework will be fundamental to our forthcoming construction \cite{HBH:scalar} of metric-stationary axisymmetric bifurcations of Kerr solving the Einstein--Klein--Gordon system, and as such, we include specializations of all of our formulas to the case of a time-periodic massive scalar field.
\end{abstract}

\section{Introduction}

The Einstein equations
\begin{equation}\label{einsteinmatter}
Ric - \frac{1}{2}g R = \mathbb{T}
\end{equation}
on a $3+1$-Lorenzian manifold $(\mathcal{M},g)$ are the fundamental equations of general relativity. Here $Ric$ and $R$ denote the Ricci curvature and scalar curvature respectively, and $\mathbb{T}$ denotes the energy-momentum tensor associated to the relevant matter model. In this paper we will be interested in setting up a framework for the problem of constructing non-vacuum, i.e., $\mathbb{T} \neq 0$, ``black hole'' solutions which are \emph{stationary}, \emph{axisymmetric}, \emph{asymptotically flat}, and possess a non-degenerate \emph{event horizon}. Since we will be concerned with the \underline{construction} of solutions, as opposed to the classification of solutions, we are free to impose any ansantz that we wish. We now turn to a description of the relevant class of spacetimes:

Let $\mathcal{M} \doteq \{(t,\phi,\rho,z)  \in \mathbb{R} \times (0,2\pi) \times \mathscr{B}\}$, where $\mathscr{B} \doteq \{(\rho,z) \in \mathbb{R}^{2} : \rho > 0\}$. We will assume that the exterior regions of our spacetimes, minus the axis of symmetry, are given by $(\mathcal{M},g)$ where the Lorentzian metrics $g$ take the form
\begin{equation}\label{eq:metric-ansatz}
g \doteq -V dt^{2} + 2W dtd\phi + X d\phi^{2} + e^{2\lambda}(d\rho^{2}+dz^{2})
\end{equation}
for suitable functions $V,W,X,\lambda : \mathscr{B}\to\mathbb{R}$. Observe that the vector fields $\Phi \doteq \partial_{\phi}$ and $T \doteq \partial_{t}$ are both Killing. We will always assume that $X>0$ (otherwise there would exist closed causal curves) and that $XV+W^{2} > 0$, which is equivalent to $g$ being a Lorentzian metric. We do \underline{not} assume that $V > 0$. Thus, we allow for the presence of an \emph{ergoregion}.

As we will see later, the ansantz~\eqref{eq:metric-ansatz} and the Einstein equations~\eqref{einsteinmatter} require us to restrict considerations here to stationary and axisymmetric energy-momentum tensors $\mathbb{T}$ (i.e., a symmetric $(0,2)$-tensor field on $\mathcal{M}$)
that satisfy
\begin{equation}\label{eq:SEtensor-orthog-directions}
\mathbb{T}(T,\partial_{\rho}) = \mathbb{T}(T,\partial_{z})  = \mathbb{T}(\Phi,\partial_{\rho}) = \mathbb{T}(\Phi,\partial_{z}) = 0.
\end{equation}

It is convenient to replace the functions $V,W,X,$ and $\lambda$ by a slightly different collection of data which reduces under the symmetries in a nice manner. The following quantities will be referred to as ``metric data:''
\begin{enumerate}
\item $X$ denotes the norm squared of the axisymmetric vector field $\Phi$: \[X \doteq g\left(\Phi,\Phi\right) > 0.\]
\item $W$ denotes the inner product between the stationary and axisymmetric vector fields $T$ and $\Phi$: \[W \doteq g\left(T,\Phi\right).\]
\item $\theta$ denotes the ``twist $1$-form'' associated to $\Phi$:
\[\theta \doteq 2i_{\Phi}\left(*\nabla\Phi_{\flat}\right).\]
\item $\sigma$ denotes the square root of the negative of the area of the parallelogram in $T\mathcal{M}$ spanned by $T$ and $\Phi$:
    \[\sigma \doteq \sqrt{XV+W^2}.\]
\item $\lambda$ is the function which appears in the conformal factor for $g$: \[\lambda \doteq \frac{1}{2}\log\left(g\left(\partial_{\rho},\partial_{\rho}\right)\right) = \frac{1}{2}\log\left(g\left(\partial_z,\partial_z\right)\right).\]
\end{enumerate}

Our first result records the manner in which the Einstein equations reduce to the $(\rho,z)$-coordinates.
\begin{theorem}\label{theo:ein-reduce}
Suppose that $(\mathcal{M},g)$ solves the Einstein equations for some energy-momentum tensor $\mathbb{T}$ satisfying \eqref{eq:SEtensor-orthog-directions}. Then the metric data satisfies the following equations on $\mathscr{B} = \{(\rho,z) \in \mathbb{R}^{2} : \rho > 0 \}$:
\begin{enumerate}
\item $X$ satisfies
\begin{align*}
& \sigma^{-1}\partial_{\rho}(\sigma\partial_{\rho}X) + \sigma^{-1}\partial_{z}(\sigma\partial_{z}X) \\
& = e^{2\lambda}(-2\mathbb{T}(\Phi,\Phi)+{\rm Tr}\left(\mathbb{T}\right)X) + \frac{(\partial_{\rho}X)^{2}+(\partial_{z}X)^{2} -\theta_{\rho}^{2}-\theta_{z}^{2}}{X}.
\end{align*}
\item $W$ satisfies
\[
\partial_{\rho}(X^{-1}W) d\rho + \partial_{z} (X^{-1}W) dz = \frac{\sigma}{X^{2}}[\theta_{\rho}dz - \theta_{z} d\rho.]
\]
\item $\theta$ satisfies
\[
d\theta =(\partial_{\rho}\theta_{z} - \partial_{z}\theta_{\rho}) d\rho\wedge dz =2\sigma^{-1}e^{2\lambda}\left(\mathbb{T}\left(\Phi,\Phi\right)W - \mathbb{T}\left(\Phi,T\right)X\right)d\rho\wedge dz,
\]
as well as
\[
\sigma^{-1}\partial_{\rho}(\sigma \theta_{\rho})+\sigma^{-1}\partial_{z}(\sigma \theta_{z}) = \frac{2\theta_{\rho}\partial_{\rho} X + 2 \theta_{z}\partial_{z}X}{X}.
\]
\item $\sigma$ satisfies
\begin{align*}
& X^{-1}e^{-2\lambda}\sigma(\partial^{2}_{\rho}\sigma + \partial^{2}_{z}\sigma) \\
& = \mathbb{T}(T-X^{-1}W\Phi,T-X^{-1}W\Phi) - X^{-2}\sigma^{2}\mathbb{T}(\Phi,\Phi) + X^{-1}\sigma^{2}{\rm Tr}(\mathbb{T}).
\end{align*}
\item $\lambda$ satisfies the following equations at the points where $|\partial\sigma|\not = 0$
\[
\partial_{\rho}\lambda = \alpha_{\rho} - \frac 12 \partial_{\rho} \log X, \qquad \partial_{z}\lambda = \alpha_{z} - \frac 12 \partial_{z}\log X,
\]
where
\begin{align*}
& \left( (\partial_{\rho}\sigma)^{2}+(\partial_{z}\sigma)^{2}\right)  \alpha_{\rho}  \\
& = \frac 12 (\partial_{\rho}\sigma) \sigma \left( \mathbb{T}(\partial_{\rho},\partial_{\rho}) - \mathbb{T}(\partial_{z},\partial_{z}) + \frac 12 X^{-2} \left[(\partial_{\rho}X)^{2} - (\partial_{z}X)^{2} + (\theta_{\rho})^{2} - (\theta_{z})^{2} \right] \right)\\
& + \partial_{\rho}\sigma (\partial^{2}_{\rho}\sigma - \partial^{2}_{z}\sigma) + \partial_{z}\sigma (\partial^{2}_{\rho,z}\sigma)\\
& + (\partial_{z}\sigma)\sigma \left[ \mathbb{T}(\partial_{\rho},\partial_{z}) + \frac 12 X^{-2}\left((\partial_{\rho}X)(\partial_{z}X) + (\theta_{\rho})(\theta_{z}) \right) \right],
\end{align*}
and
\begin{align*}
& \left( (\partial_{\rho}\sigma)^{2}+(\partial_{z}\sigma)^{2}\right) \alpha_{z} \\
& = -\frac 12 (\partial_{z}\sigma)\sigma \left( \mathbb{T}(\partial_{\rho},\partial_{\rho}) - \mathbb{T}(\partial_{z},\partial_{z}) + \frac 12 X^{-2} \left[(\partial_{\rho}X)^{2} - (\partial_{z}X)^{2} + (\theta_{\rho})^{2} - (\theta_{z})^{2} \right]  \right)\\
& - \partial_{z}\sigma (\partial^{2}_{\rho}\sigma - \partial^{2}_{z}\sigma) + \partial_{\rho}\sigma (\partial^{2}_{\rho,z}\sigma)\\
& + (\partial_{\rho}\sigma)\sigma \left[ \mathbb{T}(\partial_{\rho},\partial_{z}) + \frac 12 X^{-2} \left( (\partial_{\rho}X)(\partial_{z}X) + (\theta_{\rho})(\theta_{z}) \right)\right].
\end{align*}
Independent of the behavior of $\sigma$, $\lambda$ satisfies
\begin{align}\label{scalarcurv}
2\partial^{2}_{\rho}\lambda + 2\partial^{2}_{z}\lambda & = - \partial^{2}_{\rho}\log X - \partial^{2}_{z}\log X +  \sigma^{-1}(\partial_{\rho}^{2}\sigma + \partial^{2}_{z}\sigma)\\ \nonumber
& + e^{2\lambda}{\rm Tr}(\mathbb{T}) - \mathbb{T}(\partial_{\rho},\partial_{\rho}) - \mathbb{T}(\partial_{z},\partial_{z})\\ \nonumber
& -2 X^{-1}\left(\mathbb{T}(\Phi,\Phi) - \frac 12 {\rm Tr}\left(\mathbb{T}\right) X\right)e^{2\lambda}\\ \nonumber
& - \frac 12 X^{-2}\left[ (\partial_{\rho}X)^{2} + (\partial_{z}X)^{2} + (\theta_{\rho})^{2} + (\theta_{z})^{2} \right].
\end{align}
\end{enumerate}
Conversely, still assuming that the energy-momentum tensor satisfies \eqref{eq:SEtensor-orthog-directions}, if the metric data solves each of these equations, and $|\partial\sigma|\not = 0$ on $\mathscr{B}$, then we may recover the metric $(\mathcal{M},g)$, solving Einstein's equations \eqref{einsteinmatter}.
\end{theorem}
\begin{remark}\label{ourinfluence}We have formulated our equations so that when $\mathbb{T} = 0$ (and suitable boundary conditions are imposed) the system reduces exactly to the situation where the classical uniqueness theorems of Carter and Robinson may be applied~\cite{carter3,robinson}. In particular, when $\mathbb{T} = 0$, $\sigma$ is simply required to be a harmonic function, and the standard boundary conditions force $\sigma = \rho$. We will also have $d\theta = 0$, and thus we may write $\theta = dY$, for a unique (up to a constant) function $Y$. Then, after covering $\mathbb{R}^3$ with cylindrical coordinates $(\rho,\phi,z)$ and considering $X$ and $Y$ as functions on $\mathbb{R}^3$ which are independent of $\phi$, one easily checks that $(X,Y)$ satisfy the equations for an axisymmetric harmonic map from $\mathbb{R}^3\setminus\{\rho = 0\}$ to hyperbolic space $\mathbb{H}^2$.
\end{remark}
\begin{remark}\label{doesitextend}
Before we can claim to have constructed a stationary axisymmetric black hole spacetime, in addition to checking that $(\mathcal{M},g)$ satisfies the Einstein equations, we must also verify that $(\mathcal{M},g)$ smoothly extends to a larger spacetime $(\tilde{\mathcal{M}},\tilde g)$ which is asymptotically flat and whose boundary consists of an event horizon. Below, in Proposition \ref{canitextend}, we give sufficient conditions on the metric data for this to hold. Analogues of these conditions are already well known (cf.~\cite{carter2}) and thus we will not reproduce these in the introduction.
\end{remark}

In the application of Theorem~\ref{theo:ein-reduce} to our work~\cite{HBH:scalar} it is of fundamental importance that we may work with a first-order equation for $\lambda$, \begin{equation} \label{eq:lambda-first-order-abstract}d\lambda = \alpha - \frac 12 d\log X,\end{equation}
as opposed to the prescribed scalar curvature equation~\eqref{scalarcurv}. Essentially, this is due to the fact that if~\eqref{eq:lambda-first-order-abstract} is integrated from infinity, then the correct boundary behavior of $\lambda$ as $\rho \to 0$ (see Remark~\ref{doesitextend}) comes for ``free;'' however, in order to solve~\eqref{eq:lambda-first-order-abstract}, we must know that $d\alpha = 0$. When $\mathbb{T} = 0$ (see in Remark~\ref{ourinfluence}) the equations for $X$, $\theta$, and $\sigma$ decouple from the equations for $\lambda$, the equation for $\sigma$ (and suitable boundary conditions) force $\sigma = \rho$, and, assuming $X$ and $\theta$ solve their respective equations, it is straightforward to establish that $d\alpha = 0$ (see \cite[p.\ 916]{weinstein}).

In the non-vacuum case the situation is more complicated; our second result is the following.
\begin{theorem}\label{itallturnsouttobecompatible}
Suppose that metric data $X,W,\theta,\sigma,\lambda$ are chosen, along with a energy-momentum tensor $\mathbb{T}$ satisfying \eqref{eq:SEtensor-orthog-directions}. Assume that:
\begin{enumerate}
\item We have that $\mathbb{T} \in C^{1,\alpha}_{loc}(\mathscr{B})$, $X \in C^{2,\alpha}_{loc}(\mathscr{B})$, $W \in C^{1,\alpha}_{loc}(\mathscr{B})$, $\theta \in C^{1,\alpha}_{loc}(\mathscr{B})$, $\sigma \in C^{3,\alpha}_{loc}(\mathscr{B})$, and $\lambda \in C^{1,\alpha}_{loc}(\mathscr{B})$.
\item $X,W,\theta,\sigma$ satisfy their respective equations on $\mathscr{B}$, as listed in Theorem \ref{theo:ein-reduce}.
\item $|d\sigma| \not = 0$ on\footnote{Actually, to conclude the compatibility condition at a point, we only need to assume that $|\partial\sigma| \not = 0$ on a small neighborhood of the point under consideration.} $\mathscr{B}$.
\item If we form the metric $g$ based on the metric data, then $\mathbb{T}$ is divergence free with respect to $g$.
\end{enumerate}
Then, the $1$-form $\alpha$ arising in \eqref{eq:lambda-first-order-abstract} satisfies the following compatibility condition:
\[
d\alpha = \beta_{2}\wedge \left(d\lambda -\alpha - \frac 12 d\log X\right),
\]
for $\beta_{2}=(\beta_{2})_{\rho} d\rho + (\beta_{2})_{z}dz$ a $1$-form depending only on $\sigma$, i.e.:
\begin{align*}
(\beta_{2})_{\rho} & = \frac{1}{2} \frac{(\partial_{z}\sigma)(\partial^{3}_{\rho,z,z}\sigma - \partial^{3}_{\rho}\sigma) +(\partial_{\rho}\sigma)(\partial^{3}_{z}\sigma - \partial^{3}_{\rho,\rho,z}\sigma + 2\partial^{2}_{\rho}\sigma + 2\partial^{2}_{z}\sigma)}{(\partial_{\rho}\sigma)^{2} + (\partial_{z}\sigma)^{2}} \\
(\beta_{2})_{z} & = \frac 12 \frac{(\partial_{\rho}\sigma)(\partial^{3}_{\rho,\rho,z}\sigma-\partial^{3}_{z}\sigma) + (\partial_{z}\sigma)(\partial^{3}_{\rho}\sigma - \partial^{3}_{\rho,z,z}\sigma + 2 \partial^{2}_{\rho}\sigma + 2\partial^{2}_{z}\sigma)}{(\partial_{\rho}\sigma)^{2}+(\partial_{z}\sigma)^{2}}.
\end{align*}

Furthermore, under the above hypothesis, if $\lambda$ satisfies its first order equation \eqref{eq:lambda-first-order-abstract}, then we automatically have $\lambda \in C^{2,\alpha}_{loc}(\mathscr{B})$ and that $\lambda$ satisfies the second order equation given in Theorem \ref{theo:ein-reduce}; thus $(\mathcal{M},g)$ solves the Einstein equations \eqref{einsteinmatter}.
\end{theorem}
This compatibility condition should be thought of as quantifying the (trivial) observation that if \eqref{eq:lambda-first-order-abstract} holds, then necessarily $d\alpha = 0$. Note that for ``reasonable'' matter, the divergence free condition on $\mathbb{T}$ will automatically hold, as long as the matter fields satisfy their relevant equations with respect to the metric $g$.
\begin{remark}Implicit in our Proposition~\ref{compat-condition} below is a more general version of Theorem~\ref{itallturnsouttobecompatible} where $\mathbb{T}$ is not necessarily assumed to be divergence free; however, this extra generality is not needed for our applications to \cite{HBH:scalar}.
\end{remark}
\begin{remark}
We briefly comment on the restriction $|\partial\sigma| \not = 0$. When $\mathbb{T} = 0$, the Einstein equations imply that $\sigma$ is harmonic. This, combined with asymptotic flatness, implies that $\sigma$ has no critical points in $\mathscr{B}$. See \cite[pp.\ 915, 925]{weinstein}. For general matter, it seems possible that $\sigma$ could have critical points in $\mathscr{B}$. Note, however, that if $(\mathcal{M},g)$ is sufficiently close to Kerr, in an appropriate sense, then $\sigma$ will have no critical points (we will use this observation throughout \cite{HBH:scalar}).
\end{remark}

\subsection{The case of a massive scalar field}

With an eye towards our construction of metric-stationary black holes with scalar hair \cite{HBH:scalar}, we now specialize to the case of a massive scalar field $\Psi : \mathcal{M}\to\mathbb{C}$. Recall that in this case, the energy-momentum tensor takes the form
\begin{equation}\label{kgemt}
\mathbb{T} \doteq \text{Re}\left(\nabla\Psi\otimes \overline{\nabla\Psi}\right) - \frac{1}{2}g\left[\left|\nabla\Psi\right|^2 + \mu^2\left|\Psi\right|^2\right].
\end{equation}
Furthermore, recall that the energy-momentum tensor will be divergence free if and only if $\Psi$ satisfies the Klein--Gordon equation
\begin{equation}\label{kg}
\Box_g\Psi - \mu^2\Psi = 0.
\end{equation}

In order to fit into the framework of this paper, we assume that $\Psi$ takes the following form:
\begin{equation}\label{formscalar}
\Psi(t,\phi,\rho,z) = e^{-it\omega}e^{im\phi}\psi(\rho,z),\qquad \omega \in \mathbb{R},\ m \in \mathbb{Z},
\end{equation}
with $\psi$ real. We will refer to $m$ as the ``azimuthal number'' and $\psi$ as the ``reduced scalar field.'' One may easily check that for such a scalar field, the energy-momentum tensor satisfies the conditions given in \eqref{eq:SEtensor-orthog-directions}.

It is not hard to rewrite the Klein--Gordon equation
\[\left(\Box_g-\mu^2\right)\left(e^{-i\omega t}e^{im\phi}\psi(\rho,z)\right) = 0,\]
as
\[
\sigma^{-1}\partial_{\rho}\left(\sigma\partial_{\rho}\psi\right) + \sigma^{-1}\partial_z\left(\sigma\partial_z\psi\right) + e^{2\lambda}\sigma^{-2}X^{-1}\left(X\omega+Wm\right)^2\psi - e^{2\lambda}m^2X^{-1}\psi- e^{2\lambda}\mu^2 \psi = 0.
\]
Thus, we may characterize when $(\mathcal{M},g)$ is a solution to the Einstein--Klein--Gordon equations with scalar field $\Psi$.
\begin{theorem}
Suppose that $(\mathcal{M},g)$ solves the Einstein--Klein--Gordon equations with a scalar field $\Psi$ of the form \eqref{formscalar}. Then the metric data and reduced scalar field satisfy the following equations on $\mathscr{B} = \{(\rho,z) \in \mathbb{R}^{2} : \rho > 0 \}$:
\begin{enumerate}
\item $X$ satisfies
\[
 \sigma^{-1}\partial_{\rho}(\sigma\partial_{\rho}X) + \sigma^{-1}\partial_{z}(\sigma\partial_{z}X) = -e^{2\lambda}(2 m^{2} + \mu^{2} X) \psi^{2} + \frac{(\partial_{\rho}X)^{2}+(\partial_{z}X)^{2} -\theta_{\rho}^{2}-\theta_{z}^{2}}{X}.
\]
\item $W$ satisfies
\[
\partial_{\rho}(X^{-1}W) d\rho + \partial_{z} (X^{-1}W) dz = \frac{\sigma}{X^{2}}[\theta_{\rho}dz - \theta_{z} d\rho].
\]
\item $\theta$ satisfies
\[
d\theta =(\partial_{\rho}\theta_{z} - \partial_{z}\theta_{\rho}) d\rho\wedge dz = 2\sigma^{-1}e^{2\lambda} \left( X\omega m + W m^{2}\right) d\rho\wedge dz,
\]
as well as
\[
\sigma^{-1}\partial_{\rho}(\sigma \theta_{\rho})+\sigma^{-1}\partial_{z}(\sigma \theta_{z}) = \frac{2\theta_{\rho}\partial_{\rho} X + 2 \theta_{z}\partial_{z}X}{X}.
\]
\item $\sigma$ satisfies
\[
X^{-1}e^{-2\lambda}\sigma\left(\partial_{\rho}^2\sigma + \partial_z^2\sigma\right) = \left(\left(\omega + X^{-1}Wm\right)^2 - \sigma^2\left(\frac{\mu^2}{X} + \frac{m^2}{X^2}\right)\right)\psi^2.
\]
\item $\lambda$ satisfies the following equations at the points where $|\partial\sigma|\not = 0$
\[
\partial_{\rho}\lambda = \alpha_{\rho} - \frac 12 \partial_{\rho}\log X, \qquad \partial_{z}\lambda = \alpha_{z} - \frac 12 \partial_{z} \log X
\]
where
\begin{align*}
& \left( (\partial_{\rho}\sigma)^{2}+(\partial_{z}\sigma)^{2}\right)  \alpha_{\rho}\\
& = \frac 12 (\partial_{\rho}\sigma) \sigma \left( (\partial_{\rho}\psi)^{2} - (\partial_{z}\psi)^{2} + \frac 12 X^{-2} \left[(\partial_{\rho}X)^{2} - (\partial_{z}X)^{2} + (\theta_{\rho})^{2} - (\theta_{z})^{2} \right] \right)\\
& + \partial_{\rho}\sigma (\partial^{2}_{\rho}\sigma - \partial^{2}_{z}\sigma) + \partial_{z}\sigma (\partial^{2}_{\rho,z}\sigma)\\
& + (\partial_{z}\sigma)\sigma \left[ (\partial_{\rho}\psi)(\partial_{z}\psi) + \frac 12 X^{-2}\left((\partial_{\rho}X)(\partial_{z}X) + (\theta_{\rho})(\theta_{z}) \right) \right],
\end{align*}
and
\begin{align*}
& \left( (\partial_{\rho}\sigma)^{2}+(\partial_{z}\sigma)^{2}\right)  \alpha_{z}\\
& = -\frac 12 (\partial_{z}\sigma)\sigma \left(  (\partial_{\rho}\psi)^{2} - (\partial_{z}\psi)^{2} + \frac 12 X^{-2} \left[(\partial_{\rho}X)^{2} - (\partial_{z}X)^{2} + (\theta_{\rho})^{2} - (\theta_{z})^{2} \right]  \right)\\
& - \partial_{z}\sigma (\partial^{2}_{\rho}\sigma - \partial^{2}_{z}\sigma) + \partial_{\rho}\sigma (\partial^{2}_{\rho,z}\sigma)\\
& + (\partial_{\rho}\sigma)\sigma \left[ (\partial_{\rho}\psi)(\partial_{z}\psi) + \frac 12 X^{-2} \left( (\partial_{\rho}X)(\partial_{z}X) + (\theta_{\rho})(\theta_{z}) \right)\right].
\end{align*}
Independent of the behavior of $\sigma$, $\lambda$ satisfies
\begin{align*}
2\partial^{2}_{\rho}\lambda + 2\partial^{2}_{z}\lambda & = - \partial^{2}_{\rho}\log X - \partial^{2}_{z}\log X +  \sigma^{-1}(\partial_{\rho}^{2}\sigma + \partial^{2}_{z}\sigma)\\
& - e^{2\lambda} \mu^{2} \psi^{2} - (\partial_{\rho}\psi)^{2} - (\partial_{z}\psi)^{2} - X^{-1}\left( 2m^{2} + \mu^{2} X\right)e^{2\lambda}\psi^2\\
& - \frac 12 X^{-2}\left[ (\partial_{\rho}X)^{2} + (\partial_{z}X)^{2} + (\theta_{\rho})^{2} + (\theta_{z})^{2} \right].
\end{align*}
\item $\psi$ satisfies
\[
\sigma^{-1}\partial_{\rho}\left(\sigma\partial_{\rho}\psi\right) + \sigma^{-1}\partial_z\left(\sigma\partial_z\psi\right) + e^{2\lambda}\sigma^{-2}X^{-1}\left(X\omega+Wm\right)^2\psi - e^{2\lambda}m^2X^{-1}\psi- e^{2\lambda}\mu^2 \psi = 0.
\]
\end{enumerate}
Conversely, if the metric data and reduced scalar field solves each of these equations, and $|\partial\sigma|\not = 0$ on $\mathscr{B}$, then we may recover the metric and scalar field $(\mathcal{M},g,\Psi)$, solving the Einstein--Klein--Gordon equations.
\end{theorem}

\subsection{Acknowledgements}
YS acknowledges support from the NSF Postdoctoral Research Fellowship under award no.~1502569.

\section{Differentiable Structure and Basic Properties of the Spacetimes}
In this section we will describe the basic structure of the spacetimes we will work with. A key point will be to set up a framework so that
 \begin{enumerate}
    \item It is easy to check if a given spacetime $(\mathcal{M},g)$ extends to a larger spacetime $(\tilde{\mathcal{M}},\tilde g)$ which is asymptotically flat and possesses an event horizon.
    \item The Kerr family can be represented in a natural fashion.
 \end{enumerate}
\subsection{The Conformal Manifold with Corners $\overline{\mathscr{B}}$}
In this section we will construct a $1$-parameter family of conformal manifolds with corners $\overline{\mathscr{B}^{(\beta)}}$ on which our functions will be defined. A diagram of the manifold described below can be seen in Figure \ref{fig:conf-manifold}.

We start with a $2$-dimensional conformally flat manifold
\[\mathscr{B} \doteq \left\{(\rho,z) \in \mathbb{R}^2 : \rho > 0\right\}.\]

Let $\beta > 0$. It will be useful to consider the following three submanifolds of $\mathscr{B}$:
\begin{align*}
&\mathscr{B}^{(\beta)}_A \doteq \left\{(\rho,z) \in \mathscr{B} : \rho^2 + (z-\beta)^2 > \frac{\beta}{200}\text{, }\rho^2 + (z+\beta)^2 > \frac{\beta}{200}\text{, }\left|z\right| + \left|\rho\right| > \left(1-\frac{1}{500}\right)\beta \right\},
\\ \nonumber &\mathscr{B}^{(\beta)}_H \doteq \left\{(\rho,z) \in \mathscr{B} : \rho^2 + (z-\beta)^2 > \frac{\beta}{200}\text{, }\rho^2 + (z+\beta)^2 > \frac{\beta}{200}\text{, }\left|z\right| +\left|\rho\right| < \left(1+\frac{1}{500}\right)\beta\right\},
\\ \nonumber &\mathscr{B}^{(\beta)}_N \doteq \left\{(\rho,z) \in \mathscr{B} : \rho^2 + (z-\beta)^2 < \frac{\beta}{100}\right\},
\\ \nonumber &\mathscr{B}^{(\beta)}_S \doteq \left\{(\rho,z) \in \mathscr{B} : \rho^2 + (z+\beta)^2 < \frac{\beta}{100}\right\}.
\end{align*}

We now turn to gluing in boundaries along $\mathscr{B}^{(\beta)}_A$, $\mathscr{B}^{(\beta)}_H$, $\mathscr{B}^{(\beta)}_N$, and $\mathscr{B}^{(\beta)}_S$.

The boundaries for $\mathscr{B}^{(\beta)}_A$ and $\mathscr{B}^{(\beta)}_H$ are straightforward. Regarding $\mathscr{B}^{(\beta)}_A \cup \mathscr{B}^{(\beta)}_H$ as a subset of $\mathbb{R}^2$, we extend $\mathscr{B}^{(\beta)}_A$ and $\mathscr{B}^{(\beta)}_H$ to conformal manifolds with boundary $\overline{\mathscr{B}^{(\beta)}_A}$ and $\overline{\mathscr{B}^{(\beta)}_H}$ by gluing in the points
\[\left\{(0,z) : (z-\beta)^2 > \frac{\beta}{200}\text{, }(z+\beta)^2 > \frac{\beta}{200}\right\}.\]

Before turning to $\mathscr{B}^{(\beta)}_N$, we introduce a new set of conformally flat coordinates.
\begin{lemma}\label{schi}The coordinates $(s,\chi) \in (0,\infty) \times (0,\infty)$ defined by the formula
\[\rho \doteq s\chi,\]
\[z \doteq \frac{1}{2}\left(\chi^2 - s^2\right) + \beta,\]
represent a smooth change of coordinates and respect the conformal structure on $\mathscr{B}$, $\overline{\mathscr{B}^{(\beta)}_A}$, and $\overline{\mathscr{B}^{(\beta)}_H}$.
\end{lemma}
\begin{proof}
First of all, one easily derives the formulas
\[\chi = \sqrt{(z-\beta) + \sqrt{(z-\beta)^2+\rho^2}},\]
\[s = \sqrt{-(z-\beta) + \sqrt{(z-\beta)^2+\rho^2}}.\]
In particular, it immediately follows that $(s,\chi)$ represents a smooth change of coordinates on $\mathscr{B}$, $\overline{\mathscr{B}^{(\beta)}_A}$, and $\overline{\mathscr{B}^{(\beta)}_H}$.

To see that $(s,\chi)$ are conformally flat, we note that a straightforward calculation yields
\[d\rho^2 + dz^2 = \left(\chi^2+s^2\right)\left(ds^2 + d\chi^2\right),\]
and that $\chi^2+s^2 = 0$ if and only if $(\rho,z) = (0,\beta)$. Clearly, $(0,\beta) \not\in \overline{\mathscr{B}^{(\beta)}_A} \cup \overline{\mathscr{B}^{(\beta)}_H}\cup \mathscr{B}$.
\end{proof}
\begin{remark}The usefulness of such a coordinate system in the context of the study of stationary and axisymmetric black hole spacetimes was observed in~\cite{weinstein2}.
\end{remark}

We now extend $\mathscr{B}^{(\beta)}_N$ to a manifold with corners $\overline{\mathscr{B}^{(\beta)}_N}$ by gluing in the points \[\left\{(0,\chi) : 0 \leq \chi < \left(\frac{\beta}{25}\right)^{1/4}\right\}\qquad \text{and}\qquad \left\{(s,0) : 0 \leq s < \left(\frac{\beta}{25}\right)^{1/4}\right\}.\] Note that it is clear from Lemma~\ref{schi} that the conformal manifold with corner structure of $\overline{\mathscr{B}^{(\beta)}_N}$ is compatible with the conformal manifold with boundary structure on $\overline{\mathscr{B}^{(\beta)}_A}$ and $\overline{\mathscr{B}^{(\beta)}_H}$. We denote the corner point $(s,\chi) = (0,0)$ by $p_N$.

Let us also take the opportunity to list the formulas linking $\partial_s$ and $\partial_{\chi}$ with $\partial_{\rho}$ and $\partial_z$:
\begin{equation}\label{schitorhoz}
\partial_s = \chi\partial_{\rho} - s\partial_z,\qquad \partial_{\chi} = s\partial_{\rho} + \chi\partial_z,
\end{equation}
\begin{equation}\label{rhoztoschi}
\partial_{\rho} = \frac{\chi}{\chi^2+s^2}\partial_s + \frac{s}{\chi^2+s^2}\partial_{\chi},\qquad
\partial_z = \frac{-s}{\chi^2+s^2}\partial_s + \frac{\chi}{\chi^2+s^2}\partial_{\chi}.
\end{equation}

Finally, we turn to $\mathscr{B}^{(\beta)}_S$. Here we introduce $(s',\chi')$ coordinates by
\[\rho \doteq s'\chi',\]
\[z \doteq \frac{1}{2}\left((\chi')^2 - (s')^2\right) - \beta,\]
and, in an analogous fashion to the boundary of $\overline{\mathscr{B}^{(\beta)}_N}$, define $\overline{\mathscr{B}^{(\beta)}_S}$ by gluing in the points \[\left\{(0,\chi') : 0 \leq \chi' < \left(\frac{\beta}{25}\right)^{1/4}\right\}\qquad \text{and}\qquad\left\{(s',0) : 0 \leq s' < \left(\frac{\beta}{25}\right)^{1/4}\right\}.\]
 We denote the corner point $(s',\chi') = (0,0)$ by $p_S$.

The final conformal manifold with corners is then defined by
\[\overline{\mathscr{B}^{(\beta)}} \doteq \overline{\mathscr{B}^{(\beta)}_A} \cup \overline{\mathscr{B}^{(\beta)}_H}\cup \overline{\mathscr{B}^{(\beta)}_N} \cup \overline{\mathscr{B}^{(\beta)}_S}.\]

In order not to overburden the notation, we will often drop the $(\beta)$ from the notation and write $\overline{\mathscr{B}}$, $\overline{\mathscr{B}_A}$, $\overline{\mathscr{B}_H}$, $\overline{\mathscr{B}_N}$, and $\overline{\mathscr{B}_S}$.

Before we close the section we introduce some convenient notation for various parts of $\partial\overline{\mathscr{B}}$. When giving these definitions it is useful to observe that the $(\rho,z)$ coordinates which are initially valid in $\mathscr{B}$ extend continuously to $\overline{\mathscr{B}}$ and under this identification, $\partial\overline{\mathscr{B}} = \{\rho = 0\}$.

We define the north pole $\mathscr{A}_N$ to be the region
\[\mathscr{A}_N \doteq \left\{(\rho,z) \in \overline{\mathscr{B}} : \rho = 0 \text{ and } z > \beta\right\}.\]

We define the south pole $\mathscr{A}_S$ to be the region
\[\mathscr{A}_S \doteq \left\{(\rho,z) \in \overline{\mathscr{B}} : \rho = 0 \text{ and } z < -\beta\right\}.\]

The axis $\mathscr{A}$ is defined to be the union of the north and south pole:
\[\mathscr{A} \doteq \mathscr{A}_N \cup \mathscr{A}_S.\]

Lastly, we define the horizon $\mathscr{H}$ to be the region
\[\mathscr{H} \doteq \left\{(\rho,z) \in \overline{\mathscr{B}} : \rho = 0 \text{ and } z \in (-\beta,\beta)\right\}.\]

The highly suggestive names for these regions will be motivated in Section~\ref{statax}.

\begin{figure}[h!]
\begin{tikzpicture}
	\filldraw [opacity = .2] (0,-3) rectangle (2,2);
	\filldraw [white] (0,0) circle (.1);
	\filldraw [white] (0,-1) circle (.1);
	\draw [thick, ->] (0,.1) -- node [left] {$\mathscr{A}_N$}  (0,2.1) node [above] {$\scriptstyle \rho=0$};
	\draw [thick] (0,-.1) -- node [left] {$\mathscr{H}$}  (0,-.9);
    	\draw [thick, ->] (0,-1.1) -- node [left] {$\mathscr{A}_S$}  (0,-3.1);
	
	\node  at (-.5,0) {$\scriptstyle z=+\beta$};
	\node  at (-.5,-1) {$\scriptstyle z=-\beta$};	
	
	\draw [->] (1,-2.5) -- (1,-2) node [left] {$z$};
	\draw [->] (1,-2.5) -- (1.5,-2.5) node [below] {$\rho$};
	
	\draw  [dashed] (.05,.15) -- (3.45,1.85);
	\draw [dashed] (.05,-.15) -- (4,-.65);
	\draw [dashed] (3.95,.65) circle (1.3);

	\draw [thick, ->] (3.95,.65) -- (3.95,.65+1.3) node [below left] {$\chi$};
	\draw [thick, ->] (3.95,.65) -- (3.95 + 1.3,.65) node [below left] {$s$};
	\filldraw [opacity = .2] (3.95,.65) -- (3.95+1.3,.65) arc (0:90:1.3) -- cycle;
	\node at (3.9,.47) {$p_{N}$};
\end{tikzpicture}
\caption{A diagram of $\overline{\mathscr{B}^{(\beta)}}$. In particular, the grey shaded region (and bolded lines) on the left represents $\overline{\mathscr{B}^{(\beta)}_{A}} \cup\overline{\mathscr{B}^{(\beta)}_{H}}$, while the cutout on the right depicts $\overline{\mathscr{B}_{N}^{(\beta)}}$.}
\label{fig:conf-manifold}
\end{figure}

\subsection{Stationary and Axisymmetric Spacetimes}\label{statax}
The basic objects we will study are given by the following definition.
\begin{definition}We say that $(\mathcal{M},g)$ is a ``stationary and axisymmetric spacetime'' if
\[\mathcal{M} = \left\{(t,\phi,\rho,z) \in \mathbb{R} \times (0,2\pi) \times \mathscr{B}\right\},\]
and there exist smooth functions $V$, $W$, $X$, $\lambda : \mathscr{B} \to \mathbb{R}$ such that $X$ is positive and the expression
\begin{equation}\label{standardform}
g \doteq -Vdt^2 + 2Wdtd\phi + Xd\phi^2 + e^{2\lambda}\left(d\rho^2 + dz^2\right)
\end{equation}
defines a Lorentzian metric on $\mathcal{M}$.
\end{definition}
\begin{remark}We do \underline{not} assume that $V > 0$. However, the requirement that $g$ is Lorentzian does force
\[XV + W^2 > 0.\]
\end{remark}

It is well known from the black hole uniqueness literature that if the functions $V$, $W$, $X$, and $\lambda$ have appropriate asymptotic behaviors as $\rho \to 0$, then $(\mathcal{M},g)$ may be extended to a larger Lorentzian manifold with boundary $(\tilde{\mathcal{M}},\tilde g)$ which is asymptotically flat and has a boundary consisting of a bifurcate Killing event horizon. In the remainder of this section we will precisely describe a sufficient set of boundary conditions.

We start with the boundary conditions along $\partial\overline{\mathscr{B}_A}\cap \mathscr{A}$.
\begin{definition}\label{extend-axis} Let $\tilde{\mathscr{A}} \subset \overline{\mathscr{B}_A}$ be an open set around the axis $\mathscr{A}$.

We say that a stationary and axisymmetric spacetime $(\mathcal{M},g)$ is extendable along $\partial\overline{\mathscr{B}_A}\cap \mathscr{A}$ if
\begin{enumerate}
    \item There exists a smooth function $V_{\mathscr{A}}(\rho,z) : \tilde{\mathscr{A}} \to \mathbb{R}$ such that $V_{\mathscr{A}}\left(0,z\right) > 0$ and
        \[V\left(\rho,z\right)\big|_{\tilde{\mathscr{A}}} = V_{\mathscr{A}}\left(\rho^2,z\right).\]
    \item There exists a smooth function $W_{\mathscr{A}}(\rho,z) : \tilde{\mathscr{A}} \to \mathbb{R}$ such that \[W\left(\rho,z\right)\big|_{\tilde{\mathscr{A}}} = \rho^2W_{\mathscr{A}}\left(\rho^2,z\right).\]
    \item There exists a smooth function $X_{\mathscr{A}}(\rho,z) : \tilde{\mathscr{A}} \to \mathbb{R}$ such that $X_{\mathscr{A}}\left(0,z\right) > 0$ and
        \[X\left(\rho,z\right)\big|_{\tilde{\mathscr{A}}} = \rho^2X_{\mathscr{A}}\left(\rho^2,z\right).\]
    \item There exists a smooth function $\Sigma_{\mathscr{A}}(\rho,z) : \tilde{\mathscr{A}} \to \mathbb{R}$ such that \
        \[e^{2\lambda}\left(\rho,z\right)\big|_{\tilde{\mathscr{A}}} = X_{\mathscr{A}}\left(\rho^2,z\right) + \rho^2\Sigma_{\mathscr{A}}\left(\rho^2,z\right).\]
\end{enumerate}
\end{definition}

Next we give the boundary conditions along $\partial\overline{\mathscr{B}_H} \cap\mathscr{H}$.
\begin{definition}\label{extend-horizon}Let $\tilde{\mathscr{H}} \subset \overline{\mathscr{B}_H}$ be an open set around the horizon $\mathscr{H}$.

We say that a stationary and axisymmetric spacetime $(\mathcal{M},g)$ is extendable along $\partial\overline{\mathscr{B}_H} \cap\mathscr{H}$ if there exists $\Omega \in \mathbb{R}$ and $\kappa > 0$ such that
\begin{enumerate}
    \item There exists a smooth function $V_{\mathscr{H}}(\rho,z) : \tilde{\mathscr{H}} \to \mathbb{R}$ such that $V_{\mathscr{H}}\left(0,z\right) > 0$ and
        \[\left(V\left(\rho,z\right) -2\Omega W\left(\rho,z\right) - \Omega^2X\left(\rho,z\right)\right)\big|_{\tilde{\mathscr{H}}} = \rho^2V_{\mathscr{H}}\left(\rho^2,z\right).\]
    \item There exists a smooth function $W_{\mathscr{H}}(\rho,z) : \tilde{\mathscr{H}} \to \mathbb{R}$ such that
        \[\left(W\left(\rho,z\right) + \Omega X\left(\rho,z\right)\right)\big|_{\tilde{\mathscr{H}}} = \rho^2W_{\mathscr{H}}\left(\rho^2,z\right).\]
    \item There exists a smooth function $X_{\mathscr{H}}(\rho,z) : \tilde{\mathscr{H}} \to \mathbb{R}$ such that $X_{\mathscr{H}}\left(0,z\right) > 0$ and
        \[X\left(\rho,z\right)\big|_{\tilde{\mathscr{H}}} = X_{\mathscr{H}}\left(\rho^2,z\right).\]
    \item There exists a smooth function $\Sigma_{\mathscr{H}}(\rho,z) : \tilde{\mathscr{H}} \to \mathbb{R}$ such that
        \[e^{2\lambda}\left(\rho,z\right)\big|_{\tilde{\mathscr{H}}} = \kappa^{-2}V_{\mathscr{H}}\left(\rho^2,z\right) + \rho^2\Sigma_{\mathscr{H}}\left(\rho^2,z\right).\]
\end{enumerate}
\end{definition}

Now we turn to the region where the horizon meets the axis.
\begin{definition}\label{extend-north-pole}We say that a stationary and axisymmetric spacetime $(\mathcal{M},g)$ is extendable along $\partial\overline{\mathscr{B}_N}$ if there exists $\Omega \in \mathbb{R}$ and $\kappa > 0$ such that
\begin{enumerate}
    \item There exists a smooth function $V_N(s,\chi) : \overline{\mathscr{B}_N} \to \mathbb{R}$ such that $V_N\left(0,\chi\right) > 0$,  $V_N\left(s,0\right) > 0$, and
        \[\left(V\left(s,\chi\right) -2\Omega W\left(s,\chi\right) - \Omega^2X\left(s,\chi\right)\right)\big|_{\overline{\mathscr{B}_N}} = \chi^2V_N\left(s^2,\chi^2\right).\]
    \item There exists a smooth function $W_N(s,\chi) : \overline{\mathscr{B}_N} \to \mathbb{R}$ such that
        \[\left(W\left(s,\chi\right) + \Omega X\left(s,\chi\right)\right)\big|_{\overline{\mathscr{B}_N}} = s^2\chi^2W_N\left(s^2,\chi^2\right).\]
    \item There exists a smooth function $X_N(s,\chi) : \overline{\mathscr{B}_N} \to \mathbb{R}$ such that $X_N\left(0,\chi\right) > 0$, $X_N\left(s,0\right) > 0$, and
        \[X\left(s,\chi\right)\big|_{\overline{\mathscr{B}_N}} = s^2X_N\left(s^2,\chi^2\right).\]
    \item There exists a smooth function $\Sigma_N^{(1)}(s,\chi),\Sigma_N^{(2)} : \overline{\mathscr{B}_N} \to \mathbb{R}$ such that
        \[\left(\chi^2+s^2\right)e^{2\lambda}\left(s,\chi\right)\big|_{\overline{\mathscr{B}_N}} = X_N\left(s^2,\chi^2\right) + s^2\Sigma_N^{(1)}\left(s^2,\chi^2\right),\]
        \[\left(\chi^2+s^2\right)e^{2\lambda}\left(s,\chi\right)\big|_{\overline{\mathscr{B}_N}} = \kappa^{-2}V_N\left(s^2,\chi^2\right) + \chi^2\Sigma_N^{(2)}\left(s^2,\chi^2\right).\]
\end{enumerate}
\end{definition}
Analogously we also obtain a definition of begin ``extendable along $\partial\overline{\mathscr{B}_S}$. ''

\begin{proposition}\label{canitextend}Let $(\mathcal{M},g)$ be a stationary and axisymmetric spacetime which is extendable along $\partial\overline{\mathscr{B}_A}\cap \mathscr{A}$, $\partial\overline{\mathscr{B}_H} \cap\mathscr{H}$, $\partial\overline{\mathscr{B}_N}$, and $\partial\overline{\mathscr{B}_S}$. Then $(\mathcal{M},g)$ may be extended to a Lorentzian manfiold with corners $(\tilde{\mathcal{M}},\tilde{g})$  which is stationary and axisymmetric, and whose boundary corresponds to a bifurcate Killing event horizon.
\end{proposition}
\begin{proof}This is essentially the content of Section 10 of~\cite{carter2}. Since our use of $(s,\chi)$ coordinates in slightly different than the approach from~\cite{carter2} (we take these alternative coordinates from~\cite{weinstein2}), in Appendix~\ref{coords} we have explicitly given the coordinate systems which define $\left(\tilde{\mathcal{M}},\tilde{g}\right)$.
\end{proof}

Proposition~\ref{canitextend} naturally leads to the following definition.
\begin{definition}\label{defextend} We say that a stationary and axisymmetric spacetime $(\mathcal{M},g)$ is ``extendable to a regular black hole spacetime'' if $(\mathcal{M},g)$ satisfies the assumptions of Proposition~\ref{canitextend}.
\end{definition}

Finally, we define a suitable notion of asymptotic flatness.
\begin{definition}\label{asymflat} We say that a stationary and axisymmetric spacetime $(\mathcal{M},g)$ is asymptotically flat if in the $(t,x,y,z)$ coordinates from Appendix~\ref{axiscoor}, the metric $\tilde g$ has the smooth expansion\footnote{By smooth expansion, we mean that after every application of a $(t,x,y,z)$ derivative, the error decays one power faster.}
\[
\tilde g = \left(1+O\left(r^{-1}\right)\right)\left(-dt^{2} + dx^{2}+dy^{2}+dz^{2}\right) + O\left(r^{-2}\right)\left(dtdx + dtdy+dxdy\right).
\]
\end{definition}

\section{The Kerr Exterior is Extendable to a Regular Black Hole Spacetime}\label{kerr}
In this section we will show that Kerr exterior is extendable to a regular black hole spacetime in the sense of Definition~\ref{defextend}.

Let $M > 0$, $0 < |a| < M$, and define $\tilde r_{\pm} = M \pm\sqrt{M^2-a^2}$. Then the domain of outer communication (minus the axis of symmetry) of the Kerr spacetime of mass $M$ and angular momentum $aM$ may be covered by a Boyer--Lindquist coordinate chart\footnote{Note that we have used $\tilde r$ in place of the more commonly used $r$, because we will frequently use the definition $r^{2}=1+\rho^{2}+z^{2}$, for $\rho,z$ isothermal coordinates.} on
\[\{(t,\tilde r,\theta,\phi) \in \mathbb{R}\times (\tilde r_{+},\infty)\times (0,\pi) \times (0,2\pi)\}\]
where the metric takes the form
\begin{equation}\label{kerrmetric}
g_{a,M} = -\left( 1-\frac{2M\tilde r}{\Sigma^{2}}\right) dt^{2} - \frac{4Ma\tilde r \sin^{2}\theta}{\Sigma^{2}} dtd\phi+ \sin^{2}\theta \frac{\Pi}{\Sigma^{2}}d\phi^{2}  +\frac{\Sigma^{2}}{\Delta} d\tilde r^{2}  + \Sigma^{2} d\theta^{2}.
\end{equation}
Here
\begin{align*}
 \Delta & \doteq  \tilde r^{2} - 2M \tilde r+a^{2}\\
 \Sigma^{2} & \doteq\tilde r^{2}+a^{2}\cos^{2}\theta\\
 \Pi & \doteq (\tilde r^{2}+a^{2})^{2}-a^{2}\sin^{2}\theta \Delta.
\end{align*}
Observe that $T
\doteq \frac{\partial}{\partial t}$ and $\Phi \doteq \frac{\partial}{\partial \phi}$ are Killing vectors.

In order to bring~(\ref{kerrmetric}) into the standard form we introduce isothermal coordinates in the $(\tilde r,\theta)$ plane by
\[\rho\left(\tilde r,\theta\right) \doteq \sqrt{\Delta}\sin\theta,\]
\[z\left(\tilde r,\theta\right) \doteq  \left(\tilde r-M\right)\cos\theta.\]
After some straightforward calculations, one finds that the metric in $(t,\phi,\rho,z)$ coordinates is in the form
\begin{equation}\label{kerrmetric2}
g_{a,M} = -V_{K(a,M)} dt^{2} + 2W_{K(a,M)}dtd\phi+ X_{K(a,M)}d\phi^{2}  +e^{2\lambda_{K(a,M)}}\left(d\rho^2+dz^2\right).
\end{equation}
 with
\begin{align*}
V_{K(a,M)} & = 1 - \frac{2M\tilde r}{\Sigma^{2}}\\
W_{K(a,M)} & = - \frac{2Ma\tilde r\sin^{2}\theta}{\Sigma^{2}}\\
X_{K(a,M)} & = \sin^{2}\theta \frac{\Pi}{\Sigma^{2}}\\
e^{2\lambda_{K(a,M)}} & = \Sigma^{2}\Delta^{-1}\left(\frac{(\tilde r-M)^{2}}{\Delta} \sin^{2}\theta + \cos^{2}\theta \right)^{-1}.
\end{align*}
Often we will drop the $(a,M)$ from the notation and refer to these functions as $V_K$, $W_K$, $X_K$, and $e^{2\lambda_K}$. In this case, the dependence on the parameters $(a,M)$ is always implied. Let us also take the opportunity to note the important fact that $X_KV_K + W_K^2 = \rho^2$ (this is in fact how we found $\rho$ to begin with).

For any choice of parameters $(a,M)$ we define $\gamma^{2} \doteq M^{2}-a^{2}$. We now have

\begin{lemma}
For $M>0$ and $|a|< M$, the Kerr spacetime in $(t,\phi,\rho,z)$ coordinates is an asymptotically flat extendable stationary and axisymmetric spacetime with $\beta=\gamma$.
\end{lemma}
\begin{proof}
It is easily checked that the Jacobian of the map from $(\tilde r,\theta)$ to $(\rho,z)$ is invertible in the domain of the Boyer--Lindquist coordinates, so this is a smooth change of coordinates there. Elementary algebra allows us to solve for $\tilde r(\rho,z)$ as
\[
\tilde r(\rho,z) = M + \frac{1}{\sqrt{2}} \sqrt{ \rho^{2} + z^{2} + \gamma^{2} + \sqrt{\rho^{4} + 2\rho^{2}(z^{2}+\gamma^{2}) + (z^{2}-\gamma^{2})^{2} } }.
\]
In the $(s,\chi)$ coordinates around the north pole, we may similarly compute
\[
\tilde r(s,\chi) = M + \sqrt{\gamma^{2} + \frac 3 8 s^{2}\chi^{2} + \frac 1 8 (4\gamma - s^{2}) \left( \chi^{2}-s^{2} +  (s^{2}+\chi^{2})\sqrt{1 + \frac{\chi^{2}  + 2s^{2}+ 8\gamma}{(4\gamma - s^{2})^{2}} \chi^{2}} \right)} .
\]
A similar expression holds around the south pole.

Hence, we may extend $\tilde r$ to a smooth function on $\overline{\mathscr{B}}$.

Near the axis, we have the smooth even asymptotics\footnote{\emph{Smooth even asymptotics} is defined as follows: $f(\rho,z) = O(\rho^{2k})$ as $\rho\to 0$ means that $\frac{f(\rho,z)}{\rho^{2k}} = g(\rho^{2},z)$ for some smooth function $g(\rho,z)$ defined for $\rho\geq 0$ sufficiently small. Similarly, $f(s,\chi) = O(s^{2j}\chi^{2k})$ as $s,\chi\to 0$, will mean that $\frac{f(s,\chi)}{s^{2j}\chi^{2k}} = g(s^{2},\chi^{2})$ for some smooth function defined for $s,\chi\geq 0$ sufficiently small. }
\begin{align*}
\tilde r(\rho,z) & = M + |z| + O(\rho^{2})\\
V_{K}(\rho,z) & = \frac{z^{2}-\gamma^{2}}{(M+|z|)^{2}+a^{2}} + O(\rho^{2})\\
W_{K}(\rho,z) & = -\frac{2Ma(M+|z|)}{(z^{2}-\gamma^{2})\left( (M+|z|)^{2} + a^{2}\right)} \rho^{2} + O(\rho^{4})\\
X_{K}(\rho,z) & = \frac{ (M+|z|)^{2}+a^{2}}{z^{2}-\gamma^{2}} \rho^{2} + O(\rho^{4})\\
e^{2\lambda_{K}}(\rho,z) & = \frac{(M+|z|)^{2} + a^{2} }{ z^{2}-\gamma^{2}} + O(\rho^{2})
\end{align*}
as $\rho \to 0$. These expressions readily imply that Kerr is extendable across the axis in the sense of Definition \ref{extend-axis}.

Similarly, we may establish the following smooth even asymptotics near the horizon
\begin{align*}
\tilde r(\rho,z) & = \tilde r_{+} + \frac{\gamma}{2(\gamma^{2}-z^{2})} \rho^{2} + O(\rho^{4})\\
V_{K}(\rho,z) & = - a^{2}\frac{\gamma^{2}-z^{2}}{\gamma^{2}} \frac{1}{\Sigma^{2}} + \frac{\gamma^{2}}{\gamma^{2}-z^{2}} \frac{1}{\Sigma^{2}}\rho^{2} + O(\rho^{4})\\
W_{K}(\rho,z) & = - \frac{\gamma^{2}-z^{2}}{\gamma^{2}} \frac{2Ma\tilde r_{+}}{\Sigma^{2}} - \frac{Ma}{\gamma} \frac{1}{\Sigma^{2}} \rho^{2} + O(\rho^{4})\\
X_{K}(\rho,z) & = \frac{\gamma^{2}-z^{2}}{\gamma^{2}} \frac{4M^{2}\tilde r_{+}^{2}}{\Sigma^{2}} + \frac{\gamma^{2}-z^{2}}{\gamma^{2}} \left( \frac{4M\tilde r_{+}^{2}\gamma}{\gamma^{2}-z^{2}} - a^{2}\right) \frac{1}{\Sigma^{2}} \rho^{2} + O(\rho^{4}),
\end{align*}
as $\rho\to 0$.

Set $\Omega = \frac{a}{2M\tilde r_{+}}$ and $\kappa = \frac{\gamma}{2M\tilde r_{+}}$. From the expressions given above, we find the following smooth even asymptotics
\begin{align*}
(V_{K}-2\Omega W_{K} -\Omega^{2}X_{K})(\rho,z) & =  \frac{\gamma^{2}}{4M^{2} \tilde r_{+}^{2}}  \frac{1}{\gamma^{2}-z^{2}}\left( \tilde r_{+}^{2} + \frac{a^{2}}{\gamma^{2}} z^{2}  \right) \rho^{2} + O(\rho^{4})\\
(W_{K}+\Omega X_{K})(\rho,z) & = O(\rho^{2})\\
e^{2\lambda_{K}}(\rho,z) & = \frac{1}{\gamma^{2}-z^{2}} \left(\tilde r_{+}^{2} + \frac{a^{2}}{\gamma^{2}}z^{2} \right) + O(\rho^{2})
\end{align*}
as $\rho \to 0$, from which we may see that Kerr is extendable across the horizon in the sense of Definition \ref{extend-horizon}, with parameters $\Omega,\kappa$.

Finally, we have the following smooth even asymptotics near the north pole,
\begin{align*}
\tilde r(s,\chi) & = \tilde r_{+} + \frac{2\gamma}{4\gamma-s^{2}}\chi^{2} + O(s^{2}\chi^{4})\\
\tilde r(s,\chi) & = \tilde r_{+} + \frac{\chi^{2}}{2} +\frac{\chi^{2}}{8\gamma+2\chi^{2}} s^{2} + O(\chi^{2} s^{4})\\
V_{K}(s,\chi) & = -a^{2}\left( \frac{4\gamma^{2}}{4\gamma-s^{2}} + \frac{\chi^{2}}{2} \right)^{-1} \frac{1}{\Sigma^{2}} s^{2} + \left( \frac{4\gamma^{2}}{4\gamma-s^{2}} + \frac{\chi^{2}}{4} \right) \frac{1}{\Sigma^{2}} \chi^{2} + O(s^{2}\chi^{4})\\
W_{K}(s,\chi) & = -2Ma\tilde r_{+} \left( \frac{4\gamma^{2}}{4\gamma-s^{2}} + \frac{\chi^{2}}{4} \right)^{-1} \frac{1}{\Sigma^{2}} s^{2} - 2Ma \frac{2\gamma}{4\gamma - s^{2}} \left( \frac{4\gamma^{2}}{4\gamma-s^{2}} + \frac{\chi^{2}}{4}\right)^{-1}\frac{1}{\Sigma^{2}} s^{2}\chi^{2} + O(s^{2}\chi^{4})\\
X_{K}(s,\chi) & = 4M^{2}\tilde r_{+}^{2} \left( \frac{4\gamma^{2}}{4\gamma-s^{2}} + \frac{\chi^{2}}{4}\right)^{-1} \frac{1}{\Sigma^{2}} s^{2} + \left( \frac{16 M \gamma \tilde r_{+}^{2}}{4\gamma-s^{2}} - a^{2}s^{2}\right) \left(\frac{4\gamma^{2}}{4\gamma-s^{2}} + \frac{\chi^{2}}{4} \right)^{-1}\frac{1}{\Sigma^{2}} s^{2}\chi^{2} + O(s^{2}\chi^{4})\\
X_{K}(s,\chi) & = \left(2M\tilde r_{+} + \tilde r_{+}\chi^{2} + \frac{\chi^{4}}{4}\right)\left(\gamma + \frac{\chi^{2}}{4}\right)^{-1} s^{2}+ O(s^{4})
\end{align*}
as $s,\chi\to 0$. These expressions, along with those derived above, imply the following smooth even asymptotics
\begin{align*}
(V_{K}-2\Omega W_{K} - \Omega^{2}X_{K})(s,\chi) 
& = \frac{\gamma^{2}}{4M^{2}\tilde r_{+}^{2}} \frac{4}{4\gamma-s^{2}} \left(  2M\tilde r_{+} - \frac{a^{2}}{4\gamma^{2}} (4\gamma-s^{2})s^{2} \right)  \chi^{2}  + O(\chi^{4})\\
(W_{K}+\Omega X_{K})(s,\chi) & = O(s^{2}\chi^{2})\\
(\chi^{2}+s^{2}) e^{2\lambda_{K}}(s,\chi) & = \left( 2M \tilde r_{+} + \tilde r_{+}\chi^{2} + \frac{\chi^{4}}{4} \right)\left( \gamma + \frac{\chi^{2}}{4} \right)^{-1} + O(s^{2})\\
(\chi^{2}+s^{2}) e^{2\lambda_{K}}(s,\chi) & = \frac{4}{4\gamma-s^{2}} \left( 2M\tilde r_{+} - \frac{a^{2}}{4\gamma^{2}}(4\gamma-s^{2})s^{2}\right) + O(\chi^{2})
\end{align*}
as $s,\chi\to 0$. Putting these expressions together, we see that\footnote{Strictly speaking, we have only checked the relevant properties in a sufficiently small neighborhood of $s=\chi=0$ in $\overline{\mathscr{B}_{N}}$. This is the only thing that is needed for Proposition \ref{canitextend}; alternatively, the full statement of Definition \ref{defextend} follow from this, along with the extendibility across the axis and horizon (after a change of coordinates).} Kerr is extendible around the north pole, as in Definition \ref{extend-north-pole}. The argument for the south pole is identical.

Putting these facts together, we see that Kerr is an extendable to a regular black hole spacetime in the sense of Definition \ref{defextend}.

Finally, we turn to showing that Kerr is asymptotically flat. From the above expression, it is not hard to show that we have the smooth asymptotic falloff\footnote{Here, we will say that we have the smooth asymptotic falloff $f(\rho,z) = O(r^{-k})$ if $|\partial^{j}f|\leq C_{j}r^{-j-k}$ for $r$ sufficiently large. Recall that $r^{2}=1+\rho^{2}+z^{2}$.}

\begin{align*}
\tilde r (\rho,z) & = \sqrt{\rho^{2}+z^{2}} + M + \frac{\gamma^{2}}{4\sqrt{\rho^{2}+z^{2}}} + O(r^{-2})\\
V_{K}(\rho,z) & = 1-\frac{2M}{r} + O(r^{-2})\\
W_{K}(\rho,z) & = - \rho^{2} \frac{2Ma}{r^{3}}\left(1+O(r^{-1}) \right)\\
X_{K}(\rho,z) & = \rho^{2}\left(1+\frac{2M}{r} + O(r^{-2})\right)\\
e^{2\lambda_{K}}(\rho,z) & = 1 + \frac{2M}{r} + O(r^{-2})
\end{align*}
Using these expressions and the coordinates in Appendix \ref{axiscoor}, we may readily see that the Kerr metric is asymptotically flat in the sense of the following asymptotic falloff
\begin{align*}
g_{K} = &-\left(1-\frac {2M}{r} + O\left(r^{-2}\right)\right) dt^{2} + \left(\frac{4Ma}{r^{3}} +O\left(r^{-4}\right)\right)dt (ydx - xdy)
\\ \nonumber &+ \left( 1 + \frac{2M}{r} + O\left(r^{-2}\right)\right)(dx^{2}+dy^{2} + dz^{2}).
\end{align*}

\end{proof}

\section{Geometric Preliminaries}\label{geomprelim}
In this section we will briefly review some facts from Pseudo-Riemannian geometry which will be useful in the calculations of Section~\ref{curvcalc}.

First let's fix some notational conventions. Let $(\mathcal{N},h)$ denote an arbitrary orientable Lorentzian $4$-manifold with a fixed volume form $\epsilon$, $\bigwedge\left(\mathcal{N}\right)$ denote the space of differential forms, and $\bigwedge^k\left(\mathcal{N}\right)$ denote the space of $k$-forms on $\mathcal{N}$. Let us agree that for $\alpha_1,\cdots,\alpha_k \in \bigwedge^1(\mathcal{N})$ we have
\[\alpha_1\wedge\alpha_2\wedge\cdots\wedge \alpha_k \doteq \sum_{\sigma \in {\rm Perm}(1,\cdots,k)}{\rm sgn}(\sigma)\alpha_{\sigma_1}\otimes \alpha_{\sigma_2}\otimes \cdots\otimes \alpha_{\sigma_k}.\]
Finally, we introduce the curvature conventions:
\[R\left(X,Y\right)Z \doteq \nabla_Y\nabla_XZ - \nabla_X\nabla_YZ - \nabla_{[X,Y]}Z,\]
\[R\left(X,Y,Z,W\right) \doteq h\left(R\left(X,Y\right)Z,W\right).\]
\subsection{The Hodge Star Operator and the Electric-Magnetic Decomposition of a Differential Form}\label{emform}
We begin by reviewing the Hodge star operator, interior multiplication, and the well-known electric-magnetic decomposition of a differential form. We omit proofs as the material is standard.

\begin{definition}The \underline{Hodge star operator} $*:\bigwedge^k\left(\mathcal{N}\right) \to \bigwedge^{4-k}\left(\mathcal{N}\right)$ is the unique linear isomorphism such that for all $\alpha \in \bigwedge^k\left(\mathcal{N}\right)$ and $\beta \in \bigwedge^{4-k}(\mathcal{N})$, we have
\[\alpha\wedge *\beta = \frac{1}{k!}h\left(\alpha,\beta\right)\epsilon.\]
\end{definition}
\begin{definition}(Interior multiplication) Let $\eta \in \bigwedge^k\left(\mathcal{N}\right)$ and $K$ be a vector field on $\mathcal{N}$. Then $i_K\eta \in \bigwedge^{k-1}\left(\mathcal{N}\right)$ is defined by
\[\left(i_K\eta\right)\left(A_1,\cdots,A_{k-1}\right) \doteq \eta\left(A_1,\cdots,A_{k-1},K\right).\]
\end{definition}

The following subsets of differential forms will be useful in what follows.
\begin{definition}Let $K$ be a vector field on $\mathcal{N}$. Then we set
\[\overline{\bigwedge\nolimits^{\!2}\left(\mathcal{N}\right)}_K \doteq {\rm span}\left\{\alpha \wedge K : \alpha \in \bigwedge\nolimits^{\!1}\left(\mathcal{N}\right)\right\},\qquad \bigwedge\nolimits^{\!2}\left(\mathcal{N}\right)^{\perp}_K \doteq \left\{\eta \in \bigwedge\nolimits^{\!2}\left(\mathcal{N}\right) : i_K\eta = 0\right\}.\]
\end{definition}

\begin{lemma}(Electric-Magnetic Decomposition)\label{emdecomp} Let $K$ be a vector field on $\mathcal{N}$ such that $\left|K\right|^2 \neq 0$. Then we have an orthogonal direct sum decomposition
\begin{equation}\label{somedecomp}
\bigwedge\nolimits^{\!2}\left(\mathcal{N}\right) = \overline{\bigwedge\nolimits^{\!2}\left(\mathcal{N}\right)}_K \oplus \bigwedge\nolimits^{\!2}\left(\mathcal{N}\right)^{\perp}_K.
\end{equation}
In fact, this splitting is induced by the following explicit formula: Let $F$ be a $2$-form on $\mathcal{N}$ and define differential forms $E$ and $B$ called the ``electric field'' and ``magnetic field'' respectively by
\[E \doteq i_KF,\qquad B \doteq i_K\left(*F\right).\]

Then,
\[E\wedge K_{\flat} \in \overline{\bigwedge\nolimits^{\!2}\left(\mathcal{N}\right)}_K,\qquad *\left(B\wedge K_{\flat}\right) \in  \bigwedge\nolimits^{\!2}\left(\mathcal{N}\right)^{\perp}_K,\qquad \left|K\right|^2F = E\wedge K_{\flat} - *\left(B\wedge K_{\flat}\right).\]
We recall that the musical isomorphism $\flat$ produces a $1$-form $K_{\flat}$ defined by
\[K_{\flat}\left(A\right) \doteq h\left(K,A\right).\]
\end{lemma}
\subsection{Killing Vectorfields and their Associated Twists}
In this section will review some useful facts about Killing vector fields.

The following lemma is well-known.
\begin{lemma}\label{killem}Let $K$ be a Killing vector field. Then
 \[h\left(\nabla^2_{X,Y}K,Z\right) = R\left(Z,Y,X,K\right),\qquad \Box K_{\flat} = -i_KRic.\]
\end{lemma}
We now introduce the ``twist'' of a Killing vector field.
\begin{definition}Let $K$ be a Killing vector field on $\mathcal{N}$. Then we define a $1$-form $\vartheta \in \bigwedge^1(\mathcal{N})$, called the \underline{twist} of $K$, by
\[\vartheta \doteq 2i_K\left(*\nabla K_{\flat}\right).\]
\end{definition}
\begin{remark}In the formalism introduced in Section~\ref{emform}, we see that $\vartheta$ is proportional to the magnetic part of $\nabla K_{\flat}$ with respect to $K$.
\end{remark}
\begin{remark}As we will see later, see Remark~\ref{tvanish}, $\vartheta$ vanishes in an open set around a point $x_0$ if and only if $K$ is locally hypersurface orthogonal at $x_0$.
\end{remark}

An application of Lemma~\ref{emdecomp} and a straightforward computation shows that the electric-magnetic decomposition of $\nabla K_{\flat}$ with respect to $K$ takes the following form.
\begin{lemma}\label{emagkil}Let $K$ be a Killing vector field on $\mathcal{N}$ and $\vartheta$ be the corresponding twist $1$-form. We have
\begin{equation}\label{decompK}
\left|K\right|^2\nabla K_{\flat} = \frac{1}{2}\left(\nabla \left|K\right|^2\right)\wedge K_{\flat} - \frac{1}{2}*\left(\vartheta \wedge K_{\flat}\right).
\end{equation}
\end{lemma}

Next, we have a Bochner formula for $\left|K\right|^2$.
\begin{lemma}\label{bochner}Let $K$ be a Killing vector field on $\mathcal{N}$ such that $\left|K\right|^2$ never vanishes and $\vartheta$ be the corresponding twist $1$-form. Then,
\[\Box\left|K\right|^2 = -2Ric\left(K,K\right) + \frac{\left|\nabla\left|K\right|^2\right|^2 - \left|\vartheta\right|^2}{\left|K\right|^2}.\]
\end{lemma}
\begin{proof}First of all, Lemma~\ref{killem} yields
\begin{equation}\label{boxK}
\Box\left|K\right|^2 = 2h\left(\Box K,K\right) + 2\left|\nabla K\right|^2 = -2Ric\left(K,K\right) + 2\left|\nabla K\right|^2.
\end{equation}

Next, we have
\begin{equation}\label{normK}
\left|\nabla K\right|^2 = \left|\nabla K_{\flat}\right|^2 = \frac{1}{4}\left|K\right|^{-4}\left(\left|\nabla\left|K\right|^2\wedge K_{\flat}\right|^2 + \left|*\left(\vartheta\wedge K_{\flat}\right)\right|^2\right).
\end{equation}

We next expand the first term.
\begin{align}\label{normK1}
\frac{1}{4}\left|\nabla\left|K\right|^2\wedge K_{\flat}\right|^2 &= \frac{1}{4}\left|\nabla \left|K\right|^2\otimes K_{\flat} - K_{\flat}\otimes \nabla\left|K\right|^2\right|^2\\ \nonumber &= \frac{1}{4}\left(2\left|K\right|^2\left|\nabla\left|K\right|^2\right|^2 - 2\left(\nabla_K\left|K\right|^2\right)^2\right)\\ \nonumber  &= \frac{1}{2}\left|K\right|^2\left|\nabla\left|K\right|^2\right|^2.
\end{align}
Note that the Killing equation for $K$ implies that $\nabla_K\left|K\right|^2 = 0$.

Expanding the second term yields
\begin{align}\label{normK2}
\frac{1}{4}\left|*\left(\vartheta\wedge K_{\flat}\right)\right|^2 &= -\frac{1}{4}\left|\vartheta\wedge K_{\flat}\right|^2
\\ \nonumber &= -\frac{1}{4}\left|\vartheta\otimes K_{\flat} - K_{\flat}\otimes \vartheta\right|^2
\\ \nonumber &= -\frac{1}{4}\left(2\left|\vartheta\right|^2\left|K\right|^2 - 2\left(\vartheta\left(K\right)\right)^2\right)
\\ \nonumber &= - \frac{1}{2}\left|\vartheta\right|^2\left|K\right|^2.
\end{align}
In this calculation we used the fact that
\[\frac{1}{2}\vartheta\left(K\right) = \left(*\nabla K_{\flat}\right)\left(K,K\right) = 0.\]

Combining~(\ref{boxK}),~(\ref{normK}),~(\ref{normK1}), and~(\ref{normK2}) finishes the proof.
\end{proof}
The following formula for $d\vartheta$ will be useful in Section~\ref{theEquations}.
\begin{lemma}\label{dtheta}Let $K$ be a Killing vector field on $\mathcal{N}$ and $\vartheta$ be the associated twist. Then
\[d\vartheta = 2i_K*i_KRic.\]
\end{lemma}
\begin{proof}
We first recall (see e.g.~\cite{eelslemaire}) that for all $\eta \in \bigwedge^k\left(\mathcal{N}\right)$ we have
\[*d*\eta = -{\rm Tr}\nabla\left(\eta\right).\]
Furthermore, we have $[\mathcal{L}_K,*] = [\mathcal{L}_K,\nabla] = 0$, so that in particular,
\[\mathcal{L}_K\left(*\nabla K_{\flat}\right) = 0.\]
Finally, we recall ``Cartan's magic formula:''
\[\mathcal{L}_X\eta = di_X\eta + i_Xd\eta,\]
which holds for any vector field $X$ and $k$-form $\eta$.

We now calculate
\begin{align*}
d\vartheta &= 2d\left(i_K\left(*\nabla K_{\flat}\right)\right)
\\ \nonumber &= -2i_K\left(d*\nabla K_{\flat}\right) + 2\mathcal{L}_K\left(*\nabla K_{\flat}\right)
\\ \nonumber &= -2i_K\left(d*\nabla K_{\flat}\right)
\\ \nonumber &= 2 i_K\left(*\left(*d*\nabla K_{\flat}\right)\right)
\\ \nonumber &= -2i_K\left(*\Box K_{\flat}\right)
\\ \nonumber &= 2i_K*i_KRic. \qedhere
\end{align*}
\end{proof}
\begin{lemma}\label{divdiv}Let $\mathcal{N}$ be a $4$-dimensional Lorentzian manifold, $K$ be a Killing vector field with $\left|K\right|^2 \neq 0$, and let $\vartheta$ be the corresponding twist $1$-form. Then we have
\begin{equation}\label{divofthetheta}
{\rm div}\left(\vartheta\right) = 2\frac{h\left(\vartheta,\nabla\left|K\right|^2\right)}{\left|K\right|^2}.
\end{equation}
\end{lemma}
\begin{proof}We begin by noting that the Killing equation implies that
\[\nabla K_{\flat} = dK_{\flat},\]
and furthermore recall (see e.g.~\cite{eelslemaire}) that
\[{\rm div} = -*d*.\]
In particular,
\begin{equation}\label{div0}
{\rm div }\left(*\nabla K_{\flat}\right) = -\frac{1}{2}*d*\left(*dK_{\flat}\right) = -\frac{1}{2}*d^2K_{\flat} = 0.
\end{equation}

Then, letting $\{E_i\}$ denote a local orthonomal frame of $T\mathcal{N}$, we see that~(\ref{div0}) implies
\[{\rm div}\left(\vartheta\right) = 2\sum_i\left(*\nabla K_{\flat}\right)\left(E_i,\nabla_{E_i}K\right) = 2h\left(*\nabla K_{\flat},\nabla K_{\flat}\right).\]

Applying $*$ to~(\ref{decompK}) yields
\[\left|K\right|^2*\nabla K_{\flat} = \frac{1}{2}*\left(\nabla\left|K\right|^2\wedge K_{\flat}\right) + \frac{1}{2}\vartheta \wedge K_{\flat}.\]
Using this along with~(\ref{decompK}) yields
\begin{align}\label{divhelp}
2h\left(*\nabla K_{\flat},\nabla K_{\flat}\right) &= -\frac{1}{2}h\left(*\left(\nabla\left|K\right|^2\wedge K_{\flat}\right),*\left(\vartheta\wedge K_{\flat}\right)\right) + \frac{1}{2}h\left(\vartheta\wedge K_{\flat},\nabla \left|K\right|^2\wedge K_{\flat}\right)
\\ \nonumber &= \left|K\right|^{-4}h\left(\vartheta\wedge K_{\flat},\nabla \left|K\right|^2\wedge K_{\flat}\right)
\\ \nonumber &= \left|K\right|^{-4}h\left(\vartheta \otimes K_{\flat} - K_{\flat}\otimes \vartheta, \nabla\left|K\right|^2\otimes K_{\flat} - K_{\flat}\otimes \nabla\left|K\right|^2\right)
\\ \nonumber &= 2\left|K\right|^{-2}h\left(\vartheta,\nabla\left|K\right|^2\right).\qedhere
\end{align}
\end{proof}
\section{Curvature Calculations}\label{curvcalc}
In this section we will consider Lorentzian manifolds which have isometric $U(1)$ or $\mathbb{R}$ actions. We will establish formulas linking the Ricci curvature of the original manifold to the Ricci curvature of the quotient manifold equipped with an appropriate ``submersion'' metric.
\subsection{Curvature Decomposition for a $4$-dimensional Lorentzian Manifold Under a Spacelike $U(1)$ symmetry}\label{4to3}
In this section we let $(\mathcal{N},h)$ denote a $4$-dimensional Lorentzian manifold which admits an smooth, free, and proper isometric action by $U(1)$ with spacelike orbits. Many aspects of the presentation here mimic some calculations from~\cite{weinstein} which was concerned with the \emph{vacuum} Einstein equations on a stationary and axisymmetric spacetime.

Let $K$ denote a Killing vector field generating the symmetry, and let $\vartheta$ denote the corresponding twist $1$-form. Next, quotienting out by the orbits of the group action defines a smooth manifold $\overline{\mathcal{N}}$ and a projection map
\[\pi : \mathcal{N} \to \overline{\mathcal{N}}.\]
We now define the ``horizontal'' vector fields.
\begin{definition}We say that a vector $E$ in $T_n\mathcal{N}$ is \underline{horizontal} if $h\left(E,K\right) = 0$. We say that a vector field is horizontal if it is pointwise horizontal.
\end{definition}
Observe that it follows immediately that for each $n$, $\pi_*$ is an isomorphism from the set of horizontal vectors  in $T_n\mathcal{N}$ to $T_{\pi(n)}\overline{\mathcal{N}}$. This allows us to make the following definition.
\begin{definition}Let $\overline{E}$ be a vector field on $\overline{\mathcal{N}}$. Then the \underline{horizontal lift} of $\overline{E}$ is the unique horizontal vector field $E$ on $\mathcal{N}$ such that $\pi_*E = \overline{E}$. Any such vector field $E$ arising in this fashion is called \underline{basic}. Similarly, for any $1$-form $\overline{\eta}$ on $\overline{\mathcal{N}}$ we can define a horizontal lift $1$-form $\eta$ on $\mathcal{N}$. Again, any such $1$-form $\eta$ arising in this fashion is called basic. Finally, we introduce the convention that vector fields or $1$-forms on $\overline{\mathcal{N}}$ will come with bars, and the unbarred version will denote the horizontal lifts.
\end{definition}

The following lemma is easily proved.
\begin{lemma}\label{liezero}Let $E$ be a basic vector field. Then
\[\mathcal{L}_KE = 0.\]
\end{lemma}

This leads to the following definition and lemma.
\begin{definition}\label{horTen}We say that a tensor $\mathcal{T}$ on $\mathcal{N}$ is basic if $\mathcal{L}_K\mathcal{T}\left(E_1,\cdots,E_k,\eta_1,\cdots,\eta_l\right) = 0$ for basic vectors fields $\{E_i\}$ and $1$-forms $\{\eta_i\}$.
\end{definition}
\begin{lemma}\label{descend}Any basic tensor $\mathcal{T}$ descends to a tensor on $\overline{\mathcal{N}}$ is an unambiguous manner via the formula
\begin{equation}\label{formulaforthelift}
\mathcal{T}\left(\overline{E}_1,\cdots,\overline{E}_k,\overline{\eta}_1,\cdots,\overline{\eta}_l\right) \doteq \mathcal{T}\left(E_1,\cdots,E_k,\eta_1,\cdots,\eta_l\right),
\end{equation}
where we recall the convention that the barred vector fields and $1$-forms live on $\overline{\mathcal{N}}$ and the unbarred versions denote the horizontal lift.

More concretely, the claim is the following. Let $\overline{n} \in \overline{\mathcal{N}}$. Then, for any choice of $n$ in $\pi^{-1}\left(\overline{n}\right)$, the right hand side of~(\ref{formulaforthelift}) evaluates to the same number.
\end{lemma}
\begin{proof}Since the fiber over any point $\overline{n}$ is generated by the flow of $K$, it suffices to show that
\[K\left(\mathcal{T}\left(E_1,\cdots,E_k,\eta_1,\cdots,\eta_l\right)\right) = 0.\]
However, this follows immediately from the fact that $K$ is Killing and Lemma~\ref{liezero}.
\end{proof}

Since the metric $h$ is clearly a basic tensor, via Lemma~\ref{descend} it descends to a metric on $\overline{\mathcal{N}}$ which we denote by $\overline{h}$. In particular, $\pi : (\mathcal{N},h) \to (\overline{\mathcal{N}},\overline{h})$ will now form a \emph{Riemannian submersion} in the sense of~\cite{oneil}. Our immediate goal is to express the Ricci curvature of $(\mathcal{N},h)$ in terms of $\left|K\right|^2$, $\vartheta$, and the Ricci curvature of $(\overline{\mathcal{N}},\overline{h})$. Lemmas~\ref{bochner} and~\ref{dtheta} have already shown that $i_KRic$ is determined by $\left|K\right|^2$ and $\vartheta$, so it remains to study $Ric\left(E,F\right)$ for horizontal tensors $E$ and $F$.

Following~\cite{oneil} we introduce the fundamental tensor $D$ on $\overline{\mathcal{N}}$.
\begin{definition}\label{Atensor}We define a tensor $D$ of type $(0,2)$ on $\mathcal{N}$ by
\[D\left(E,F\right) \doteq \left|K\right|^{-2}h\left(\nabla_EF,K\right).\]
It follows immediately from the fact that $K$ is Killing, Lemma~\ref{liezero}, and the Lie derivative product rule that $D$ is basic and hence descends to $\overline{\mathcal{N}}$.
\end{definition}

The following straightforward lemma is proved in~\cite{oneil}.
\begin{lemma}\label{levicivitbar}Let $\overline{E}$ and $\overline{F}$ be vector fields on $\overline{\mathcal{N}}$. Then $\mathcal{L}_K\left(\nabla_EF\right) = 0$, and thus we may define a connection on $\overline{\mathcal{N}}$ by
\[\overline{\nabla}_{\overline{E}}\overline{F} \doteq \pi_*\left(\nabla_EF\right).\]
In fact, this connection is equal to the Levi--Civita connection for $(\overline{N},\overline{h})$.

\end{lemma}
Since the tensor $D$ measures the projection of $\nabla_EF$ onto the kernel of $\pi_*$, Lemma~\ref{levicivitbar} implies that $D$ may be interpreted as a measure of the difference between $\overline{\nabla}$ and $\nabla$.

In the next lemma we derive formulas linking $D$ to the Lie bracket and $\vartheta$
\begin{lemma}Let $D$ be the tensor from Definition~\ref{Atensor} and $E$ and $F$ be basic vector fields. We have
\[D\left(E,F\right) = \frac{1}{2}\left|K\right|^{-2}h\left([E,F],K\right) = \left|K\right|^{-4}*\left(\vartheta\wedge K_{\flat}\right)\left(E,F\right).\]
\end{lemma}
\begin{proof}The proof of the first equality can be found in~\cite{oneil}.

The second equality follows from
\begin{align*}
D\left(E,F\right) &= -\left|K\right|^{-2}\left(\nabla K_{\flat}\right)\left(E,F\right)
\\ \nonumber &= -\left|K\right|^{-4}\left[\left(\nabla\left|K\right|^2\wedge K_{\flat}\right)\left(E,F\right) - *\left(\vartheta\wedge K_{\flat}\right)\left(E,F\right)\right]
\\ \nonumber &= \left|K\right|^{-4}*\left(\vartheta\wedge K_{\flat}\right)\left(E,F\right).\qedhere
\end{align*}
\end{proof}
\begin{remark}\label{tvanish}Note that this lemma implies that $\vartheta$ is the obstruction to $K$ being locally hypersurface orthogonal.
\end{remark}

The fundamental equations of a Riemannian submersion (see~\cite{oneil}) give us
\begin{proposition}Let $\overline{E}$, $\overline{F}$, $\overline{G}$, and $\overline{H}$ be vector fields on $\overline{\mathcal{N}}$ and let $\overline{R}$ denote the curvature tensor of $(\overline{\mathcal{N}},\overline{h})$. Then, keeping in mind that $K$ Killing implies that $R$ is a basic tensor, we have
\begin{equation}\label{curvsub}
\overline{R}\left(\overline{E},\overline{F},\overline{G},\overline{H}\right) =
\end{equation}
\[R\left(\overline{E},\overline{F},\overline{G},\overline{H}\right) + \left|K\right|^2\left(2D\left(\overline{E},\overline{F}\right)D\left(\overline{G},\overline{H}\right) - D\left(\overline{F},\overline{G}\right)D\left(\overline{E},\overline{H}\right) - D\left(\overline{G},\overline{E}\right)D\left(\overline{F},\overline{H}\right)\right).\]
\end{proposition}

Next, we want to trace~(\ref{curvsub}) in order to derive an equation for the Ricci curvature of $(\overline{\mathcal{N}},\overline{h})$. We will break the calculation up into a few lemmas. First of all, we record the easily proved fact that various tensors of interest are basic.
\begin{lemma}\label{theyarehorizontal}The function $\left|K\right|^2$ and tensors $\nabla\left|K\right|^2$, $\nabla K$, $\nabla^2K$, and $\vartheta$ are all basic and thus may be considered to be functions or tensors on $\overline{\mathcal{N}}$.
\end{lemma}

The following calculation will be important.
\begin{lemma}\label{important}Let $E$ and $F$ be two basic vector fields. Then
\[h\left(i_E*\left(\vartheta\wedge K_{\flat}\right),i_F*\left(\vartheta\wedge K_{\flat}\right)\right) = -\left|K\right|^2\left[\left|\vartheta\right|^2h\left(E,F\right) - \vartheta\left(E\right)\vartheta\left(F\right)\right].\]
\end{lemma}
\begin{proof}This calculation is most easily done in a local orthonormal frame. Let $\left\{e_0,e_1,e_2,\frac{K}{\left|K\right|}\right\}$ be a local orthonormal frame in $T\mathcal{N}$ and $\left\{\omega^0,\omega^1,\omega^2,\frac{K_{\flat}}{\left|K\right|}\right\}$ be the dual local orthonormal frame in $T^*\mathcal{N}$. Let us furthermore assume that $\omega^0\wedge \omega^1 \wedge \omega^2 \wedge K_{\flat}$ is positively oriented and that $h\left(\omega^0,\omega^0\right) = -1$. Next, for $l$, $p$ in $\{0,1,2\}$ we introduce the notation $\gamma\left(l,p\right)$ for the unique number in $\{0,1,2\}\setminus \{l,p\}$.

An easy calculation then shows that it suffices to prove that
\begin{equation}\label{ans1}
h\left(i_{e_j}*\left(\vartheta\wedge K_{\flat}\right),i_{e_j}*\left(\vartheta\wedge K_{\flat}\right)\right) = \left|K\right|^2\sum_{i\neq j}^2\left(\vartheta_i\right)^2h\left(\omega^{\gamma(i,j)},\omega^{\gamma(i,j)}\right),
\end{equation}
and, for $j \neq k$,
\begin{align}\label{ans2}
h\left(i_{e_j}*\left(\vartheta\wedge K_{\flat}\right),i_{e_k}*\left(\vartheta\wedge K_{\flat}\right)\right) &= \left|K\right|^2\vartheta_k\vartheta_j.
\end{align}

We start by writing
\[\vartheta\wedge K_{\flat} = \sum_{i=0}^2\vartheta_i\omega^i\wedge K_{\flat},\]
and
\[*\left(\vartheta \wedge K_{\flat}\right) = \left|K\right|\sum_{i=0}^2\left(-1\right)^i\vartheta_ih\left(\omega^i,\omega^i\right)\omega^0\wedge\cdots\wedge \hat{\omega^i}\wedge\cdots \wedge \omega^2.\]

Continuing,
\[i_{e_j}*\left(\vartheta\wedge K_{\flat}\right) = \left|K\right|\sum_{i\neq j}^2\left(-1\right)^i{\rm sgn}\left(\gamma(i,j),j\right)h\left(\omega^i,\omega^i\right)\vartheta_i\omega^{\gamma(i,j)}.\]

We now easily see that
\begin{equation*}
h\left(i_{e_j}*\left(\vartheta\wedge K_{\flat}\right),i_{e_j}*\left(\vartheta\wedge K_{\flat}\right)\right) = \left|K\right|^2\sum_{i\neq j}^2\left(\vartheta_i\right)^2h\left(\omega^{\gamma(i,j)},\omega^{\gamma(i,j)}\right).
\end{equation*}
This establishes~(\ref{ans1}).

Now we consider the case when $j \neq k$. We have
\begin{align}\label{intertwistinner}
h&\left(i_{e_j}*\left(\vartheta\wedge K_{\flat}\right),i_{e_k}*\left(\vartheta\wedge K_{\flat}\right)\right)
\\ \nonumber &= \left|K\right|^2h\left(\sum_{i\neq j}^2\left(-1\right)^ih\left(\omega^i,\omega^i\right){\rm sgn}\left(\gamma(i,j),j\right)\vartheta_i\omega^{\gamma(i,j)},\sum_{i\neq k}^2\left(-1\right)^ih\left(\omega^i,\omega^i\right){\rm sgn}\left(\gamma(i,k),k\right)\vartheta_i\omega^{\gamma(i,k)}\right).
\end{align}

Then, we observe
\begin{align}\label{dothesum}
\sum_{i\neq j}^2\left(-1\right)^i{\rm sgn}\left(\gamma(i,j),j\right)\vartheta_i\omega^{\gamma(i,j)} = \left(-1\right)^k&h\left(\omega^k,\omega^k\right){\rm sgn}\left(\gamma(k,j),j\right)\vartheta_k\omega^{\gamma\left(k,j\right)} \\ \nonumber &+ \left(-1\right)^{\gamma(k,j)}h\left(\omega^{\gamma(k,j)},\omega^{\gamma(k,j)}\right){\rm sgn}\left(k,j\right)\vartheta_{\gamma\left(k,j\right)}\omega^k.
\end{align}
Plugging~(\ref{dothesum}) into~(\ref{intertwistinner} yields
\begin{align*}
h&\left(i_{e_j}*\left(\vartheta\wedge K_{\flat}\right),i_{e_k}*\left(\vartheta\wedge K_{\flat}\right)\right)
\\ \nonumber &= \left|K\right|^2(-1)^k(-1)^jh\left(\omega^j,\omega^j\right)h\left(\omega^k,\omega^k\right)h\left(\omega^{\gamma(k,j)},\omega^{\gamma(k,j)}\right){\rm sgn}\left(\gamma(k,j),j\right){\rm sgn}\left(\gamma(j,k),k\right)\vartheta_k\vartheta_j
\\ \nonumber &= \left|K\right|^2\vartheta_k\vartheta_j.
\end{align*}
This establishes~(\ref{ans2}) and thus concludes the proof.
\end{proof}

\begin{lemma}\label{sortofimportant}Let $E$ and $F$ be basic vector fields. Then
\[R\left(E,K,K,F\right) = \frac{1}{2}\nabla^2_{E,F}\left|K\right|^2 - \frac{1}{4}\left|K\right|^{-2}\left(\nabla_E\left|K\right|^2\right)\left(\nabla_F\left|K\right|^2\right) +\frac{1}{4}\left|K\right|^{-2}\left[\left|\vartheta\right|^2h\left(E,F\right) - \vartheta\left(E\right)\vartheta\left(F\right)\right].\]
\end{lemma}
\begin{proof}We start with the following identity:
\begin{align*}
\nabla^2_{E,F}\left|K\right|^2 &= 2h\left(\nabla^2_{E,F}K,F\right) + 2h\left(\nabla_EK,\nabla_FK\right).
\\ \nonumber &= 2R\left(E,K,K,F\right) + 2h\left(\nabla_EK,\nabla_F,K\right).
\end{align*}
Thus,
\begin{equation}\label{curvKK1}
R\left(E,K,K,F\right) = \frac{1}{2}\nabla^2_{E,F}\left|K\right|^2 - h\left(\nabla_EK,\nabla_FK\right).
\end{equation}
Now we expand
\[\nabla_EK = -i_E\left(\nabla K_{\flat}\right) = -\frac{1}{2}\left|K\right|^{-2}i_E\left(\nabla\left|K\right|^2\wedge K_{\flat} - *\left(\vartheta\wedge K_{\flat}\right)\right) = \frac{1}{2}\left|K\right|^{-2}\nabla_E\left|K\right|^2K_{\flat} + \frac{1}{2}\left|K\right|^{-2}i_E*\left(\vartheta\wedge K_{\flat}\right).\]
Using Lemma~\ref{important}, we then find that
\begin{align}\label{curvKK2}
h\left(\nabla_EK,\nabla_FK\right) &= \frac{1}{4}\left|K\right|^{-2}\left(\nabla_E\left|K\right|^2\right)\left(\nabla_F\left|K\right|^2\right) + \frac{1}{4}\left|K\right|^{-4}h\left(i_E*\left(\vartheta\wedge K_{\flat}\right),i_F*\left(\vartheta\wedge K_{\flat}\right)\right)
\\ \nonumber &= \frac{1}{4}\left|K\right|^{-2}\left(\nabla_E\left|K\right|^2\right)\left(\nabla_F\left|K\right|^2\right) -\frac{1}{4}\left|K\right|^{-2}\left[\left|\vartheta\right|^2h\left(E,F\right) - \vartheta\left(E\right)\vartheta\left(F\right)\right].
\end{align}
Combining this with~(\ref{curvKK1}) completes the proof.
\end{proof}

We are now ready to compute the Ricci curvature of $\overline{h}$.
\begin{proposition}\label{ricdecompose}
\begin{equation}
\overline{Ric}\left(\overline{E},\overline{F}\right) =
\end{equation}
\begin{equation*}
Ric\left(\overline{E},\overline{F}\right) + \frac{1}{2}\left|K\right|^{-2}\overline{\nabla}^2_{\overline{E},\overline{F}}\left|K\right|^2 - \frac{1}{4}\left|K\right|^{-4}\left(\overline{\nabla}_{\overline{E}}\left|K\right|^2\right)\left(\overline{\nabla}_{\overline{F}}\left|K\right|^2\right)-\frac{1}{2}\left|K\right|^{-4}\left[\left|\vartheta\right|^2\overline{h}\left(\overline{E},\overline{F}\right) -\vartheta\left(\overline{E}\right)\vartheta\left(\overline{F}\right)\right].
\end{equation*}
\end{proposition}
\begin{proof}First of all, since $\left|K\right|^2$ is basic, it may be considered to be a function on both $\overline{\mathcal{N}}$ and $\mathcal{N}$. With this abuse of notation in mind, for any vector field $Z$ on $\mathcal{N}$, we have
\[Z\left|K\right|^2 = \left(\pi_*Z\right)\left|K\right|^2.\]
In particular, letting $E$ and $F$ be horizontal lifts of $\overline{E}$ and $\overline{F}$, we have
\begin{align}\label{hesssame}
\nabla^2_{E,F}\left|K\right|^2 &= E\left(F\left|K\right|^2\right) - \nabla_EF\left(\left|K\right|^2\right)
\\ \nonumber &= \overline{E}\left(\overline{F}\left|K\right|^2\right) - \pi_*\left(\nabla_EF\right)\left(\left|K\right|^2\right)
\\ \nonumber &= \overline{E}\left(\overline{F}\left|K\right|^2\right) - \overline{\nabla}_EF\left(\left|K\right|^2\right)
\\ \nonumber &= \overline{\nabla}^2_{\overline{E},\overline{F}}\left|K\right|^2.
\end{align}

Using~(\ref{hesssame}) and additionally Lemma~\ref{theyarehorizontal} we see that the proposition will follow if we establish
\begin{equation}
\overline{Ric}\left(\overline{E},\overline{F}\right) =
\end{equation}
\begin{equation*}
Ric\left(E,F\right) + \frac{1}{2}\left|K\right|^{-2}\nabla^2_{E,F}\left|K\right|^2 - \frac{1}{4}\left|K\right|^{-4}\left(\nabla_{E}\left|K\right|^2\right)\left(\nabla_{F}\left|K\right|^2\right)-\frac{1}{2}\left|K\right|^{-4}\left[\left|\vartheta\right|^2h\left(E,F\right) -\vartheta\left(E\right)\vartheta\left(F\right)\right].
\end{equation*}

Letting $\{G_i\}$ denote the horizontal lift of a local orthonormal frame $\{\overline{G}_i\}$ in $T\overline{\mathcal{N}}$, tracing the formula~(\ref{curvsub}) and using Lemma~\ref{sortofimportant} yields
\begin{align}\label{curvsub2}
\overline{Ric}\left(\overline{E},\overline{F}\right) &= Ric\left(E,F\right) + \left|K\right|^{-2}R\left(E,K,K,F\right) - \sum_{i=1}^33\left|K\right|^2D\left(E,G_i\right)D\left(G_i,F\right)
\\ \nonumber &= Ric\left(E,F\right) + \left|K\right|^{-2}\frac{1}{2}\nabla^2_{E,F}\left|K\right|^2 - \frac{1}{4}\left|K\right|^{-4}\left(\nabla_E\left|K\right|^2\right)\left(\nabla_F\left|K\right|^2\right) \\ \nonumber &\qquad \qquad +\frac{1}{4}\left|K\right|^{-4}\left[\left|\vartheta\right|^2h\left(E,F\right) - \vartheta\left(E\right)\vartheta\left(F\right)\right]
\\ \nonumber &\qquad \qquad +3\left|K\right|^{-6}h\left(i_{E}D,i_{F}D\right)
\\ \nonumber &= Ric\left(E,F\right) + \frac{1}{2}\left|K\right|^{-2}\nabla^2_{E,F}\left|K\right|^2 - \frac{1}{4}\left|K\right|^{-4}\left(\nabla_E\left|K\right|^2\right)\left(\nabla_F\left|K\right|^2\right)
\\ \nonumber & \qquad \qquad -\frac{1}{2}\left|K\right|^{-4}\left[\left|\vartheta\right|^2h\left(E,F\right) - \vartheta\left(E\right)\vartheta\left(F\right)\right].\qedhere
\end{align}

\end{proof}

Finally, we can get rid of the term proportional to the Hessian of $|K|^2$ by carrying out a conformal transformation.

\begin{proposition}\label{hatRic}Define a new metric $\hat h \doteq \left|K\right|^2 \overline{h}$. Then we have
\begin{align*}
\hat{Ric}\left(\overline{E},\overline{F}\right) &= Ric\left(\overline{E},\overline{F}\right) + \left|K\right|^{-2}Ric\left(K,K\right)\overline{h}\left(\overline{E},\overline{F}\right) +  \\ \nonumber&\qquad \frac{1}{2}\left|K\right|^{-4}\left[\left(\hat{\nabla}_{\overline{E}}\left|K\right|^2\right)\left(\hat{\nabla}_{\overline{F}}\left|K\right|^2\right) + \vartheta\left(\overline{E}\right)\vartheta\left(\overline{F}\right)\right].
\end{align*}
\end{proposition}
\begin{proof}If we define a new metric $\hat h \doteq f\overline{h}$ on $\overline{\mathcal{N}}$, the Ricci curvature transforms according to the general formula
\[\hat{Ric} = \overline{Ric} - \frac{1}{f}\overline{\nabla}^2f - \frac{1}{2f}\left(\Box_{\overline{h}} f \right)\overline{h} + \frac{3}{4f^2}\overline{\nabla}f\otimes \overline{\nabla}f + \frac{1}{4f^2}\left|\overline{\nabla}f\right|^2\overline{h}.\]

If we set $f = \left|K\right|^2$ and use Proposition~\ref{ricdecompose}, then we obtain
\[\hat{Ric}\left(\overline{E},\overline{F}\right) = \]
\[Ric\left(\overline{E},\overline{F}\right) + \frac{1}{2}\left|K\right|^{-4}\left(\overline{\nabla}_{\overline{E}}\left|K\right|^2\right)\left(\overline{\nabla}_{\overline{F}}\left|K\right|^2\right) - \frac{1}{2}\left|K\right|^{-4}\left[\left|\vartheta\right|^2\overline{h}\left(\overline{E},\overline{F}\right) - \vartheta\left(\overline{E}\right)\vartheta\left(\overline{F}\right)\right]\]
 \[- \frac{1}{2}\left|K\right|^{-2}\left(\Box_{\overline{h}}\left|K\right|^2\right)\overline{h}\left(\overline{E},\overline{F}\right) + \frac{1}{4}\left|K\right|^{-4}\left|\overline{\nabla}\left|K\right|^2\right|^2\overline{h}\left(\overline{E},\overline{F}\right).\]

Next, using Lemma~\ref{bochner}, and the fact that $\nabla_KK_{\flat} = -\frac{1}{2}\nabla\left|K\right|^2$ (this follows immediately from the Killing equation) we compute
\begin{align*}
\Box_{\overline{h}}\left|K\right|^2 &= \Box_h\left|K\right|^2 - \frac{1}{\left|K\right|^2}\nabla^2_{K,K}\left|K\right|^2
\\ \nonumber &= -2Ric\left(K,K\right) + \left|K\right|^{-2}\left[\left|\overline{\nabla}\left|K\right|^2\right|^2 - \left|\vartheta\right|^2\right] + \left|K\right|^{-2}\left(\nabla_KK\right)\left|K\right|^2
\\ \nonumber &= -2Ric\left(K,K\right) + \left|K\right|^{-2}\left[\frac{1}{2}\left|\overline{\nabla}\left|K\right|^2\right|^2 - \left|\vartheta\right|^2\right].
\end{align*}
Plugging this into the formula above finishes the proof.

\end{proof}
\subsection{Prescribing the Ricci Curvature for a $3$-dimensional Lorentzian Manifold Under an Orthogonal Timelike Translational Symmetry with a Conformally Flat Quotient}\label{3to2}
In this section we will consider a $3$-dimensional Lorentzian manifold $\left(\hat{\mathcal{N}},\hat{h}\right)$ which admits an smooth, free, and proper isometric action by $\mathbb{R}$ with \emph{globally hypersurface orthogonal} timelike orbits. We let $U$ denote the infinitesimal generator of the symmetry, and define $Q \doteq \sqrt{-\hat{h}\left(U,U\right)}$. Quotienting out by the orbits of the group action defines a smooth manifold $\slashed{N}$ and a projection map
\[\pi : \hat{\mathcal{N}} \to \slashed{N}.\]
Since $U$ is assumed to be globally hypersurface orthogonal, we may define (non-uniquely) an inclusion map
\[i: \slashed{N} \to \hat{\mathcal{N}}\]
which will satisfy $\pi \circ i = {\rm Id}$. Fixing the inclusion map $i$ we now identify $\slashed{N}$ and $i\left(\slashed{N}\right)$. Since $\slashed{N}$ is now a submanifold of $\hat{\mathcal{N}}$ we get an induced metric $\slashed{h}$. We now make the further \emph{assumption} that $(\slashed{N},\slashed{h})$ is conformally flat in that there exists a single coordinate chart $(x,y) \in (0,\infty)\times \mathbb{R}$ covering $\slashed{N}$ and that there exists a function $f : (0,\infty) \times \mathbb{R} \to \mathbb{R}$ so that in the $(x,y)$ coordinate system
\[\slashed{h} = e^{2f}\left(dx^2 + dy^2\right),\footnote{We use coordinates in $(0,\infty)\times \mathbb{R}$ rather than $\mathbb{R}^2$ since this will fit naturally with other aspects of our gauge. However, the calculations of this section will all be local, so the range of $x$ and $y$ will not in fact enter into any of our arguments.} \]
and thus
\[\hat{h} = -Q^2dU^2 + e^{2f}\left(dx^2+dy^2\right).\]
Note that $Q>0$ on $\hat{\mathcal{N}}$.

The goal of this section is to understand to what extent the Ricci curvature of $\hat{h}$ may be freely prescribed. General vector fields, the connection, and curvature quantities associated to $\hat{h}$ will be hatted and general vector fields, the connection, and curvature quantities associated to $\slashed{h}$ will be slashed.

We begin with the (trivial) observation that $\slashed{N}$ is totally geodesic and hence that the second fundamental form vanishes.

An immediate corollary of the vanishing of the second fundamental form and Codazzi's equation (see~\cite{doCarmo}) is the following.
\begin{corollary}\label{hatRiemU}Let $\slashed{E}$, $\slashed{F}$, and $\slashed{G}$ be vector fields in $T\slashed{N}$. Then
\[\hat{R}\left(U,\slashed{E},\slashed{F},\slashed{G}\right) = 0.\]
In particular, $\hat{Ric}\left(U,\partial_x\right) = \hat{Ric}\left(U,\partial_y\right) = 0$.
\end{corollary}

In order to determine further constraints of $\hat{Ric}$ it will be useful to fully calculate $\hat{\nabla}$. The following lemma follows from a straightforward calculation.
\begin{lemma}\label{derU}We have
\begin{align*}
\hat{\nabla}_UU = Qe^{-2f}\left[\partial_xQ\partial_x + \partial_yQ\partial_y\right],&\qquad \hat{\nabla}_xU = Q^{-1}(\partial_xQ)U,
\\ \nonumber \hat{\nabla}_yU = Q^{-1}(\partial_yQ)U,&\qquad \hat{\nabla}_U\partial_x = Q^{-1}(\partial_xQ)U,
\\ \nonumber \hat{\nabla}_U\partial_y = Q^{-1}(\partial_yQ) U,&\qquad \hat{\nabla}_{\partial_x}\partial_x = \partial_xf\partial_x - \partial_yf\partial_y,
\\ \nonumber \hat{\nabla}_{\partial_y}\partial_y = - \partial_xf\partial_x + \partial_yf\partial_y ,&\qquad \hat{\nabla}_{\partial_x}\partial_y = \partial_yf\partial_x + \partial_xf\partial_y.
\end{align*}
\end{lemma}
Now we are ready to derive a fundamental equation linking $\hat{\nabla}^2Q$ to $\hat{R}$.
\begin{proposition}\label{curvtohess}We have
\[\hat{R}\left(U,\slashed{E},U,\slashed{F}\right) = Q\hat{\nabla}^2_{\slashed{E},\slashed{F}}Q.\]
\end{proposition}
\begin{proof}We have
\[\hat{R}\left(\slashed{E},U,U,\slashed{F}\right) = \hat{h}\left(\hat{\nabla}_U\hat{\nabla}_{\slashed{E}}U - \hat{\nabla}_{\slashed{E}}\hat{\nabla}_UU - \hat{\nabla}_{[U,\slashed{E}]}U,\slashed{F}\right).\]

First of all, $[U,\slashed{E}] = 0$, so that in particular,
\[-\hat{h}\left(\hat{\nabla}_{[U,\slashed{E}]}U,\slashed{F}\right) = 0.\]

Next, using Lemma~\ref{derU} we compute
\[\hat{h}\left(\hat{\nabla}_U\hat{\nabla}_{\slashed{E}}U,\slashed{F}\right) = \hat{h}\left(\hat{\nabla}_U\left(Q^{-1}\left(\hat{\nabla}_{\slashed{E}}Q\right)U\right),\slashed{F}\right) = \left(\hat{\nabla}_{\slashed{E}}Q\right)\left(\hat{\nabla}_{\slashed{F}}Q\right).\]

Then, using Lemma~\ref{derU} again, we compute

\[-\hat{h}\left(\hat{\nabla}_{\slashed{E}}\hat{\nabla}_UU,\slashed{F}\right) = -\hat{h}\left(\hat{\nabla}_{\slashed{E}}\left(Q\left(\hat{\nabla} Q\right)\right),\slashed{F}\right) = -\left(\hat{\nabla}_{\slashed{E}}Q\right)\left(\hat{\nabla}_{\slashed{F}}Q\right) - Q\hat{\nabla}_{\slashed{E},\slashed{F}}Q.\]

Adding these three equations together yields the result.
\end{proof}

Keeping in mind that $\hat{\nabla} = \slashed{\nabla}$ (since the second fundamental form of $\slashed{N}$ vanishes), tracing the statement of Proposition~\ref{curvtohess} immediately yields the following corollary.
\begin{corollary}\label{deltQRic}We have
\[Q\slashed{\Delta}Q = \hat{Ric}\left(U,U\right).\]
Equivalently,
\[Q\left(\partial_x^2Q + \partial_y^2Q\right) = e^{2f}\hat{Ric}\left(U,U\right).\]
\end{corollary}
\begin{remark}\label{Qequalx}In the case when $\hat{Ric}\left(U,U\right) = 0$ we see that $Q$ is harmonic with respect to the flat metric $dx^2 + dy^2$; in particular, there is no reference to $f$, and $Q$ satisfies a linear PDE. However, when $\hat{Ric}\left(U,U\right) \neq 0$, then the above Corollary implies that $Q$ is a solution to a non-linear PDE with a right hand side depending on $\hat{Ric}\left(U,U\right)$ and $f$.
\end{remark}

The final fundamental constraints on $\hat{Ric}$ arise from the requirement that the Einstein tensor be divergence free.

\begin{lemma}\label{bianchi-equations}The requirement that the Einstein tensor be divergence free is
\[\hat {\rm div}\left( \hat{Ric} - \frac{1}{2}\hat{R}\hat{h}\right) = 0.\]
Using that $U$ is Killing, a straightforward calculation shows that the $U$ component of the identity is vacuous. In contrast, the $\partial_x$ and $\partial_y$ components take the following form when written out in coordinates.
\begin{align}\label{ein1}
\frac{1}{2}e^{-2f}&\partial_x\hat{Ric}(\partial_x,\partial_x) -\frac{1}{2} e^{-2f}\partial_x\hat{Ric}\left(\partial_y,\partial_y\right) + \frac{1}{2}Q^{-2}\partial_x\hat{Ric}(U,U) \\ \nonumber &+ e^{-2f}\partial_y\hat{Ric}(\partial_x,\partial_y)
 +e^{-2f}Q^{-2}(\partial_xQ)\hat{Ric}(\partial_x,\partial_x) + Q^{-1}e^{-2f}(\partial_yQ)\hat{Ric}(\partial_x,\partial_y) = 0,
\end{align}
\begin{align}\label{ein2}
\frac{1}{2}e^{-2f}&\partial_y\hat{Ric}(\partial_y,\partial_y) -\frac{1}{2} e^{-2f}\partial_y\hat{Ric}\left(\partial_x,\partial_x\right) + \frac{1}{2}Q^{-2}\partial_y\hat{Ric}(U,U) \\ \nonumber &+ e^{-2f}\partial_x\hat{Ric}(\partial_x,\partial_y)
 +e^{-2f}Q^{-2}(\partial_yQ)\hat{Ric}(\partial_y,\partial_y) + Q^{-1}e^{-2f}(\partial_xQ)\hat{Ric}(\partial_x,\partial_y) = 0.
\end{align}
\end{lemma}
\begin{proof}
\begin{align*}
\left(\hat\nabla \hat{Ric}\right)\left(U,U,\partial_x\right) &= -\hat{Ric}\left(\hat\nabla_UU,\partial_x\right) - \hat{Ric}\left(U,\hat\nabla_U\partial_x\right)
\\ \nonumber &= -e^{-2f}Q(\partial_xQ) \hat{Ric}\left(\partial_x,\partial_x\right) - e^{-2f}Q(\partial_yQ)Ric\left(\partial_x,\partial_y\right) - Q^{-1}(\partial_xQ)\hat{Ric(U,U)}.
\end{align*}
Next,
\begin{align*}
\left(\hat\nabla \hat{Ric}\right)\left(\partial_x,\partial_x,\partial_x\right) &= \partial_x\hat{Ric}\left(\partial_x,\partial_x\right) - 2\hat{Ric}\left(\hat\nabla_{\partial_x}\partial_x,\partial_x\right)
\\ \nonumber &= \partial_x\hat{Ric}\left(\partial_x,\partial_x\right) - 2\partial_xf\hat{Ric}\left(\partial_x,\partial_x\right) + 2\partial_yf\hat{Ric}\left(\partial_x,\partial_y\right).
\end{align*}
Then,
\begin{align*}
\left(\hat{\nabla}\hat{Ric}\right)\left(\partial_y,\partial_y,\partial_x\right) &= \partial_y\hat{Ric}\left(\partial_x,\partial_y\right) - \hat{Ric}\left(\hat{\nabla}_{\partial_y}\partial_y,\partial_x\right) - \hat{Ric}\left(\partial_y,\hat{\nabla}_{\partial_y}\partial_x\right)
\\ \nonumber &= \partial_y\hat{Ric}(\partial_x,\partial_y) - \partial_yf\hat{Ric}\left(\partial_y,\partial_x\right) + \partial_xf\hat{Ric}(\partial_x,\partial_x) - \partial_yf\hat{Ric}(\partial_x,\partial_y) - \partial_xf\hat{Ric}(\partial_y,\partial_y).
\end{align*}
Finally,
\begin{align*}
\partial_x\hat{R} &= \partial_x\left(-Q^{-2}\hat{Ric}(U,U) + e^{-2f}\hat{Ric}\left(\partial_x,\partial_x\right) + e^{-2f}\hat{Ric}\left(\partial_y,\partial_y\right)\right)
\\ \nonumber &= -Q^{-3}\partial_xQ\hat{Ric}(U,U)  - Q^{-2}\partial_x\hat{Ric}(U,U) - 2e^{-2f}\partial_xf\hat{Ric}(\partial_x,\partial_x)
 \\ \nonumber &\qquad \qquad \qquad + e^{-2f}\partial_x\hat{Ric}(\partial_x,\partial_x) - 2e^{-2f}\partial_xf\hat{Ric}(\partial_y,\partial_y) + e^{-2f}\partial_x\hat{Ric}(\partial_y,\partial_y).
\end{align*}

Putting everything together yields
\begin{align*}
0 &= -Q^{-2}(\hat{\nabla}\hat{Ric})(U,U,\partial_x) + e^{-2f}(\hat{\nabla}\hat{Ric})(\partial_x,\partial_x,\partial_x) + e^{-2f}(\hat{\nabla}\hat{Ric})(\partial_y,\partial_y,\partial_x) - \frac{1}{2}\partial_x\hat{R}
\\ \nonumber &= \frac{1}{2}e^{-2f}\partial_x\hat{Ric}(\partial_x,\partial_x) -\frac{1}{2} e^{-2f}\partial_x\hat{Ric}\left(\partial_y,\partial_y\right) + \frac{1}{2}Q^{-2}\partial_x\hat{Ric}(U,U) + e^{-2f}\partial_y\hat{Ric}(\partial_x,\partial_y)
\\ \nonumber &\qquad \qquad +e^{-2f}Q^{-2}\partial_xQ\hat{Ric}(\partial_x,\partial_x) + Q^{-1}e^{-2f}\partial_yQ\hat{Ric}(\partial_x,\partial_y)
\end{align*}
Reversing the roles of $x$ and $y$ yields
\begin{align*}
0 &= -Q^{-2}(\hat{\nabla}\hat{Ric})(U,U,\partial_y) + e^{-2f}(\hat{\nabla}\hat{Ric})(\partial_y,\partial_y,\partial_y) + e^{-2f}(\hat{\nabla}\hat{Ric})(\partial_x,\partial_x,\partial_y) - \frac{1}{2}\partial_y\hat{R}
\\ \nonumber &= \frac{1}{2}e^{-2f}\partial_y\hat{Ric}(\partial_y,\partial_y) -\frac{1}{2} e^{-2f}\partial_y\hat{Ric}\left(\partial_x,\partial_x\right) + \frac{1}{2}Q^{-2}\partial_y\hat{Ric}(U,U) + e^{-2f}\partial_x\hat{Ric}(\partial_y,\partial_x)
\\ \nonumber &\qquad \qquad +e^{-2f}Q^{-2}\partial_yQ\hat{Ric}(\partial_y,\partial_y) + Q^{-1}e^{-2f}\partial_xQ\hat{Ric}(\partial_x,\partial_y)\qedhere
\end{align*}
\end{proof}

\begin{remark}Note the possibly surprising fact that~(\ref{ein1}) and~(\ref{ein2}) do not involve $\nabla f$.
\end{remark}

As is well known, the scalar curvature $\slashed{R}$ determines the conformal factor $f$ via the Liouville equation. We record the Liouville equation for later use.
\begin{lemma}\label{liouville-f-eqn}
The conformal factor $f$ satisfies
\[
2(\partial^{2}_{x}f + \partial^{2}_{y}f) = Q^{-1}(\partial^{2}_{x}Q+\partial^{2}_{y}Q) - \hat{Ric}(\partial_{x},\partial_{x}) - \hat{Ric}(\partial_{y},\partial_{y})
\]
\end{lemma}
It is also useful to record the derivative of the Liouville equation.
\begin{lemma}\label{deriv-liouville-f-eqn}
The conformal factor $f$ satisfies
\begin{align*}
2\partial_{x}(\partial^{2}_{x}f + \partial^{2}_{y}f) & = \partial_{x}(Q^{-1}(\partial^{2}_{x}Q + \partial^{2}_{y}Q)) - \partial_{x}\hat Ric(\partial_{x},\partial_{x}) - \partial_{x}\hat Ric(\partial_{y},\partial_{y})\\
2\partial_{y}(\partial^{2}_{x}f + \partial^{2}_{y}f) & = \partial_{y}(Q^{-1}(\partial^{2}_{x}Q + \partial^{2}_{y}Q)) - \partial_{y}\hat Ric(\partial_{x},\partial_{x}) - \partial_{y}\hat Ric(\partial_{y},\partial_{y})
\end{align*}
\end{lemma}

The Liouville equation will be useful in a few places below, but our use of isothermal coordinates on $\slashed{N}$ will in fact allow us to derive a first order equation for $f$ which will end up being much easier to analyze analytically than the Liouville equation would have been, which we now do.

\begin{proposition}The conformal factor $f$ satisfies
\[(\partial_xf)(\partial_xQ) - (\partial_yf)(\partial_yQ) = \frac{Q}{2}\left(\hat{Ric}\left(\partial_x,\partial_x\right) - \hat{Ric}\left(\partial_y,\partial_y\right)\right) + \partial_x^2Q - \partial_y^2Q,\]
\[(\partial_xf)(\partial_yQ) + (\partial_yf)(\partial_xQ) = \partial_{x,y}^2Q + Q\hat{Ric}\left(\partial_x,\partial_y\right).\]
\end{proposition}
\begin{proof}We have
\begin{align*}
\hat{Ric}\left(\partial_x,\partial_x\right) &= -Q^{-2}\hat{R}\left(\partial_x,U,\partial_x,U\right) + e^{-2f}\hat{R}\left(\partial_x,\partial_y,\partial_x,\partial_y\right)
\\ \nonumber &= -Q^{-1}\hat{\nabla}^2_{\partial_x,\partial_x}Q + e^{-2f}\hat{R}\left(\partial_x,\partial_y,\partial_x,\partial_y\right)
\\ \nonumber &= -Q^{-1}\partial_x^2Q + Q^{-1}\left[(\partial_xf)(\partial_xQ)-(\partial_yf)(\partial_yQ)\right] + e^{-2f}\hat{R}\left(\partial_x,\partial_y,\partial_x,\partial_y\right).
\end{align*}
Reversing the roles of $x$ and $y$ yields
\[\hat{Ric}\left(\partial_y,\partial_y\right) = -Q^{-1}\partial_y^2Q + Q^{-1}\left[(\partial_yf)(\partial_yQ)-(\partial_xf)(\partial_xQ)\right] + e^{-2f}\hat{R}\left(\partial_x,\partial_y,\partial_x,\partial_y\right).\]
Subtracting the two equations then yields
\[(\partial_xf)(\partial_xQ) - (\partial_yf)(\partial_yQ) = \frac{Q}{2}\left(\hat{Ric}\left(\partial_x,\partial_x\right) - \hat{Ric}\left(\partial_y,\partial_y\right)\right) + \partial_x^2Q - \partial_y^2Q.\]
Next,
\begin{align*}
\hat{Ric}\left(\partial_x,\partial_y\right) &= -Q^{-2}\hat{R}\left(\partial_x,U,\partial_y,U\right)
\\ \nonumber &= -Q^{-1}\hat{\nabla}^2_{x,y}Q
\\ \nonumber &= -Q^{-1}\partial_{x,y}^2Q + Q^{-1}\left[(\partial_xf)(\partial_yQ) + (\partial_yf)(\partial_xQ)\right].
\end{align*}
Rearranging yields
\[(\partial_xf)(\partial_yQ) + (\partial_yf)(\partial_xQ) = \partial_{x,y}^2Q + Q\hat{Ric}\left(\partial_x,\partial_y\right).\qedhere\]
\end{proof}
\begin{corollary}\label{first-order-eqns-for-f} Assuming that $|\partial Q| \not = 0$, we can rewrite the above
\[\partial_xf = \frac{\frac{1}{2}(\partial_xQ)Q\left(\hat{Ric}(\partial_x,\partial_x) - \hat{Ric}(\partial_y,\partial_y)\right) + \partial_xQ(\partial_x^2Q - \partial_y^2Q) + (\partial_yQ)(\partial^2_{x,y}Q) + (\partial_yQ)Q\hat{Ric}(\partial_x,\partial_y)}{(\partial_xQ)^2 + (\partial_yQ)^2},\]
\[\partial_yf = \frac{-\frac{1}{2}(\partial_yQ)Q\left(\hat{Ric}(\partial_x,\partial_x) - \hat{Ric}(\partial_y,\partial_y)\right) - \partial_yQ(\partial_x^2Q - \partial_y^2Q) + (\partial_xQ)(\partial^2_{x,y}Q) + (\partial_xQ)Q\hat{Ric}(\partial_x,\partial_y)}{(\partial_xQ)^2 + (\partial_yQ)^2}.\]
\end{corollary}

\subsection{Compatibility of the first order equations} In this section, we will consider $Q$ as a fixed function. We would like to investigate to what degree we may prescribe $\hat{Ric}(\partial_{x},\partial_{x})$, $\hat{Ric}(\partial_{x},\partial_{y})$, and $\hat{Ric}(\partial_{y},\partial_{y})$ for the metric
\[
\hat h = - Q^{2}dU^{2} + e^{2f}(dx^{2}+dy^{2})
\]
by solving the first order equations for $f$ from Corollary \ref{first-order-eqns-for-f}.

 Recalling that
\[
\hat{Ric}(U,U) = e^{-2f} Q(\partial_{x}^{2}Q+\partial^{2}_{y}Q),
\]
the Bianchi equations derived in Lemma \ref{bianchi-equations} can be rewritten as the requirement that the $1$-form $\mathfrak{B} = \mathfrak{B}_{x}dx + \mathfrak{B}_{y}dy$ vanishes, where
\begin{align}\label{ein1-noU}
\mathfrak{B}_{x} & \doteq \frac{1}{2}\partial_x\hat{Ric}(\partial_x,\partial_x) -\frac{1}{2} \partial_x\hat{Ric}\left(\partial_y,\partial_y\right) + \partial_y\hat{Ric}(\partial_x,\partial_y)
 +Q^{-2}(\partial_xQ)\hat{Ric}(\partial_x,\partial_x)   \\ \nonumber & + Q^{-1}(\partial_yQ)\hat{Ric}(\partial_x,\partial_y) - Q^{-1}(\partial^{2}_{x}Q+\partial^{2}_{y}Q) \partial_{x}f  + \frac{1}{2} Q^{-2}\partial_x(Q(\partial^{2}_{x}Q+\partial^{2}_{y}Q)) ,
\end{align}
\begin{align}\label{ein2-noU}
\mathfrak{B}_{y} & \doteq \frac{1}{2}\partial_y\hat{Ric}(\partial_y,\partial_y) -\frac{1}{2} \partial_y\hat{Ric}\left(\partial_x,\partial_x\right) + \partial_x\hat{Ric}(\partial_x,\partial_y)
 +Q^{-2}(\partial_yQ)\hat{Ric}(\partial_y,\partial_y)  \\ \nonumber & + Q^{-1}(\partial_xQ)\hat{Ric}(\partial_x,\partial_y)   - Q^{-1}(\partial^{2}_{x}Q+\partial^{2}_{y}Q)\partial_{y}f + \frac{1}{2}Q^{-2}\partial_y(Q(\partial^{2}_{x}Q+\partial^{2}_{y}Q)) .
\end{align}
An important observation is that by rewriting the Bianchi equations in this manner, we have cancelled the $e^{2f}$ factors. Moreover, since we are considering $Q$ as a fixed function, the pointwise value of these expressions depend (affine) linearly on the pointwise values of $\hat{Ric},\partial\hat{Ric}$, and $\partial f$. This observation plays an essential role in our analysis of the compatibility condition.

Now, based on Corollary \ref{first-order-eqns-for-f} we define a $1$-form
\begin{equation}\label{alphadef}
\alpha= \alpha_{x}dx+ \alpha_{y}dy
\end{equation}
by the expressions
\[\alpha_{x} \doteq \frac{\frac{1}{2}(\partial_x Q )Q \left(\hat{Ric}(\partial_x ,\partial_x ) - \hat{Ric}(\partial_y ,\partial_y )\right) + \partial_x Q (\partial_x ^2Q  - \partial_y ^2Q ) + (\partial_y Q )(\partial^2_{x ,y }Q ) + (\partial_y Q )Q \hat{Ric}(\partial_x ,\partial_y )}{(\partial_x Q )^2 + (\partial_y Q )^2},\]
\[\alpha_{y} \doteq \frac{-\frac{1}{2}(\partial_y Q )Q \left(\hat{Ric}(\partial_x ,\partial_x ) - \hat{Ric}(\partial_y ,\partial_y )\right) - \partial_y Q (\partial_x ^2Q  - \partial_y ^2Q ) + (\partial_x Q )(\partial^2_{x ,y }Q ) + (\partial_x Q )Q \hat{Ric}(\partial_x ,\partial_y )}{(\partial_x Q )^2 + (\partial_y Q )^2} .\]
Then, we can write the first order equations for $f$ succinctly as $df = \alpha$. Thus, if we consider $Q$, $\hat{Ric}(\partial_{x},\partial_{x})$,$\hat{Ric}(\partial_{x},\partial_{y})$, and $\hat{Ric}(\partial_{y},\partial_{y})$ to be chosen, then the (local) obstruction to solving the first-order equation is $d\alpha = 0$.

We are thus led to investigate the relationship between the Bianchi equations and $d\alpha$ in the form of the following compatibility condition.
\begin{proposition}\label{compat-condition}
Suppose that $\mathcal{U}\subset\slashed{\mathcal{N}}$ is a simply connected neighborhood of some point $p \in \slashed{\mathcal{N}}$. Assume that $Q \in C^{3,\alpha}(\mathcal{U})$ is a fixed positive function with $|\partial Q| \not = 0$ on $\mathcal{U}$. Then, there are $1$-forms $\beta_{1},\beta_{2} \in C^{0,\alpha}(\mathcal{U};T^{*}\slashed{\mathcal{N}})$ with the following property.

Consider $R_{xx},R_{xy},R_{yy},f \in C^{1,\alpha}(\mathcal{U})$. Then, if we substitute $R_{xx}$ for $\hat{Ric}(\partial_{x},\partial_{x})$, $R_{xy}$ for $\hat{Ric}(\partial_{x},\partial_{y})$ and so on in the above expressions for $\alpha$ and $\mathfrak{B}$, then
\[
d\alpha = \beta_{1}\wedge \mathfrak{B} + \beta_{2}\wedge (df - \alpha).
\]
Moreover, these requirements uniquely determine the $1$-forms $\beta_{1},\beta_{2}$; indeed their coefficients at $p$ are algebraic functions for the $3$-jet of $Q$ at $p$. We will determine $\beta_{2}$ explicitly below.

Finally, if $R_{xx},R_{xy},R_{yy},f$ are chosen so that $\alpha$ and $\mathfrak{B}$ satisfy $\mathfrak{B} = df - \alpha = 0$ on $\mathcal{U}$, then $f$ is in $C^{2,\alpha}(\mathcal{U})$ and it satisfies the Liouville equation from Lemma \ref{liouville-f-eqn}, i.e.
\[
2(\partial^{2}_{x}f + \partial^{2}_{y}f) = Q^{-1}(\partial^{2}_{x}Q+\partial^{2}_{y}Q) - R_{xx} - R_{yy}.
\]
\end{proposition}
It is possible to prove this by a direct but somewhat long computation. Instead, we give a proof based on the (affine) linearity of certain expressions, along with a dimension count.

\begin{proof}Throughout, we will consider $p \in \mathcal{U}\subset \slashed{\mathcal{N}}$ and $Q \in C^{3,\alpha}(\mathcal{U})$ with $|\partial Q|\not =0$ on $\mathcal{U}$ all to be fixed data.

We define two vector spaces, $\mathbf{V}_{\text{curv}}$ and $\mathbf{V}_{\text{conf}}$, which will contain data pertaining to $\hat{Ric}$ and its derivative, and certain derivatives of $f$, respectively. First of all, we set
\[
\mathbf{V}_{\text{curv}} \doteq Sym^{2}(T^{*}_{p}\slashed{\mathcal{N}}) \oplus (T^{*}_{p}\slashed{\mathcal{N}} \otimes Sym^{2}(T^{*}_{p}\slashed{\mathcal{N}})).
\]
This should be thought of as a metric containing the $1$-jet of the Ricci curvature at $p$ (only in the $\partial_{x}$ and $\partial_{y}$ directions). Note that given a metric
\[
\hat h = -Q^{2}dU^{2} + e^{2f}(dx^{2} + dy^{2}),
\]
we may naturally associate it to an element of $\mathbf{V}_{\text{curv}}$ as
\[
\hat h \mapsto \left( \hat{Ric}(\partial_{x},\partial_{x}) dx^{2} +\dots + \hat{Ric}(\partial_{y},\partial_{y})dy^2, \partial_{x}\hat{Ric}(\partial_{x},\partial_{x}) dx\otimes(dx^{2}) + \dots + \partial_{y}\hat{Ric}(\partial_{y},\partial_{y}) dy\otimes (dy^{2}) \right).
\]
Note that we only keep the $x,y$ components of the curvature and its derivative.

Secondly, we define
\[
\mathbf{V}_{\text{conf}} \doteq T^{*}_{p}\slashed{\mathcal{N}} \oplus Sym^{2}(T^{*}_{p}\slashed{\mathcal{N}}) \oplus \bigwedge\nolimits^{\!2} T^{*}_{p}\slashed{\mathcal{N}} \oplus \mathbb{R} \oplus T^{*}_{p}\slashed{\mathcal{N}}.
\]
This should (loosely) be thought of as containing information about the $2$-jet of the conformal factor $f$ at $p$ (except for $f(p)$) as well as the $1$-jet of the Laplacian of $f$ at $p$. Note that if $f \in C^{2,\alpha}(\mathcal{U})$, then we can naturally map it to the following element of $\mathbf{V}_{\text{conf}}$:
\[
f\mapsto (df(p),D^{2}f(p),0,(\partial^{2}_{x}f + \partial^{2}_{y}f)(p),d(\partial^{2}_{x}f + \partial^{2}_{y}f)(p))\in\mathbf{V}_{\text{conf}}.
\]
Let us emphasize that $\mathbf{V}_{\text{conf}}$ does not record the value of $f(p)$ under this map.

Note that the relationship between $f$ and the element of $\mathbf{V}_{\text{conf}}$ is not as straightforward as the description of $\mathbf{V}_{\text{curv}}$ above. Firstly, $\mathbf{V}_{\text{conf}}$ contains a slot for an ``antisymmetric component of the Hessian.'' Secondly, $\mathbf{V}_{\text{conf}}$ contains information about the Laplacian twice (as it contains a slot for all of the second derivatives of $f$, as well as a slot for the Laplacian term). Essentially, the goal of the proof of Proposition \ref{compat-condition} is to show that a certain element of $\mathbf{V}_{\text{conf}}$ (naturally constructed out of the $1$-form $\alpha$ and other data) has no ``antisymmetric component of the Hessian,'' and the two versions of the Laplacian agree.

Based on this discussion, we define a compatibility map $c: \mathbf{V}_{\text{conf}}\to \bigwedge^{2}T^{*}_{p}\slashed{\mathcal{N}} \oplus \mathbb{R}$ by
\[
(f_{x}dx + f_{y}dy, f_{xx}dx^{2} + 2f_{xy}dxdy + f_{yy}dy^{2}, \tilde f_{xy} dx\wedge dy, \hat f, \hat f_{x}dx + \hat f_{y}dy) \mapsto (\tilde f_{xy}dx\wedge dy, f_{xx} + f_{yy} - \hat f).
\]
for \emph{real numbers} $f_{x},f_{y},\dots,\hat f_{y}$. In words, the kernel of $c$ is precisely the elements of $\mathbf{V}_{\text{conf}}$ that have ``vanishing anti-symmetric Hessian'' and whose both possible Laplacians agree. We set \[\widetilde{\mathbf{V}}_{\text{conf}} \doteq \ker c.\]
We note for later use that $c$ is surjective (by inspection).

Now, we define a map $F: \mathbf{V}_{\text{curv}}\to \mathbf{V}_{\text{conf}}$, which essentially bundles up the data of $\alpha$ (and its derivative) as well as the Liouville equation (and its derivative), both of which we can think of giving us information about some part of the $3$-jet of $f$ from the $1$-jet of $\hat{Ric}$. More precisely, given an element of $\mathbf{V}_{\text{curv}}$, using~\eqref{alphadef} we may define a $1$-form $\alpha$ (recall that we are considering $Q \in C^{3,\alpha}(\mathcal{U})$ to be fixed). Moreover, we may formally compute the derivative of each coefficient of $\alpha$ and then use the element of $\mathbf{V}_{\text{curv}}$ (along with $Q$) to find a ``Jacobian'' matrix for $\alpha$, i.e.
\[
J_{\alpha} \doteq \left(\begin{matrix}
\partial_{x}\alpha_{x} & \partial_{y}\alpha_{x}\\\partial_{x}\alpha_{y} & \partial_{y}\alpha_{y}
\end{matrix}\right).
\]
Furthermore, we may use Lemmas \ref{liouville-f-eqn} and \ref{deriv-liouville-f-eqn} to use the data from $\mathbf{V}_{\text{curv}}$ (along with $Q$) to formally define $\Delta_{f},d(\Delta_{f})$ at $p$. Putting this together, $F$ will map the chosen element of $\mathbf{V}_{\text{curv}}$ to
\[
(\alpha, Sym^{2} J_{\alpha},{\textstyle\bigwedge}^{2}J_{\alpha},\Delta_{f},d(\Delta_{f}))\in \mathbf{V}_{\text{conf}}.
\]
The map $F$ is not linear, but it is \emph{affine}\footnote{
For our purposes, it is convenient to define an affine space to be a pair $(\mathbf{V}',\mathbf{V})$ where $\mathbf{V}$ is a (real) vector space, and $\mathbf{V}'$ is a subset so that $\mathbf{V}'-v'$ is a linear subspace for some (and hence any) $v' \in\mathbf{V}'$. Vector spaces are (trivially) affine linear spaces. We will usually drop $\mathbf{V}$ from the notation for an affine space. A map $a : \mathbf{V}_{1}'\to\mathbf{V}_{2}'$ is said to be affine linear if there is $v_{1}\in V'$ so that the function $a'(v) := a(v-v_{1}') + a(v_{1}')$ defines a linear map $a': \mathbf{V}_{1}'-v_{1}' \to \mathbf{V}_{2}'-a(v_{1}')$.}
linear.\footnote{Note that this is precisely why we are considering $Q$ as fixed. If we tried to bundle $Q$ up into a map of this sort, then we would completely lose any sort of linearity, as $\alpha$ depends on the $1$-jet of $Q$ in a truly nonlinear way.} It is convenient to define $\hat F \doteq F \circ \proj_{\mathbf{V}_{\text{curv}}} : \mathbf{V}_{\text{curv}} \oplus T^{*}_{p}\slashed{\mathcal{N}} \to \mathbf{V}_{\text{conf}}$.

We claim that $F$ is surjective. To see this, we first consider the restricted map sending $R_{ab}$ to $(\alpha,\Delta_{f})$. Then, if we consider the coefficient of $R_{xx},R_{xy},R_{yy}$ in $\Delta_{f},\alpha_{x},\alpha_{y}$, we find the matrix
\[
\bordermatrix{
~ & \Delta_{f} & \alpha_{x} & \alpha_{y} \cr
R_{xx} & -1 & \frac 12 Q |\partial Q|^{-2} \partial_{x}Q & -\frac 12Q |\partial Q|^{-2}\partial_{y}Q\cr
R_{xy} & 0 & Q |\partial Q|^{-2} \partial_{y}Q & Q |\partial Q|^{-2}\partial_{x}Q \cr
R_{yy} & -1 & -\frac 12 Q |\partial Q|^{-2} \partial_{x}Q & \frac 12 Q |\partial Q|^{-2} \partial_{y}Q \cr}
\]
which has determinant $-Q \not = 0$. Hence, we see that we may freely prescribe $\Delta_{f},\alpha_{x},\alpha_{y}$ by choosing $R_{xx},R_{xy},R_{yy}$. Thinking of this choice as fixed, we may then use an identical argument to prescribe $\partial_{x}(\Delta_{f})$, $\partial_{x}\alpha_{x}$, and $\partial_{x}\alpha_{y}$ in terms of $\partial_{x}R_{ab}$, and then so on for the $y$-derivative (the reason for this is that we may consider any terms in $\partial_{x}(\Delta_{f})$, $\partial_{x}\alpha_{a}$ where the $\partial_{x}$-derivative does not fall on a Ricci term as being part of the affine shift; this is because we have already chosen the $R_{ab}$ terms and $Q$ is fixed, as usual. This shows that the corresponding matrix is identical). Next, note that by a dimension count:
\[
\dim \mathbf{V}_{\text{curv}} = 3 + 2 \times 3 = 9 = 2 + 3 + 1 + 1 + 2 = \dim\mathbf{V}_{\text{conf}},
\]
we in fact conclude that $F$ is bijective.

Conversely, we define a map $R :\widetilde{\mathbf{V}}_{\text{conf}}\to \mathbf{V}_{\text{curv}}\oplus T^{*}_{p}\slashed{\mathcal{N}}$ as follows. Given
\[
(f_{x}dx+f_{y}dy, f_{xx}dx^{2}+2f_{xy}dxdy +f_{yy}dy^{2},0,f_{xx}+f_{yy},\hat f_{x}dx +\hat f_{y} dy) \in \widetilde{\mathbf{V}}_{\text{conf}},
\]
there is a unique polynomial function $f(x,y)$ so that $f(p) = 0$, $\partial_{a} f(p) = f_{a}$ for $a \in \{x,y\}$, $\partial^{2}_{a,b}f(p) = f_{ab}$ for $a,b\in\{x,y\}$, $\partial^{3}_{aaa}f(p) = \hat f_{a}$ for $a\in\{x,y\}$, and finally forcing all other derivatives to vanish.\footnote{Note that there was some freedom in choosing the third derivatives of $f$. In fact, we only need to choose the third derivatives to satisfy $d(\Delta f) = \hat f_{x}dx + \hat f_{y}dy$, so we have chosen to also require all mixed third derivatives vanish.} 
Then, using this function $f$, we consider the metric $-Q^{2}dU^{2} + e^{2f}(dx^{2}+dy^{2})$ and compute its Ricci curvature and derivative of Ricci curvature (all in the $\partial_{x},\partial_{y}$ directions) at $p$. As usual, we can naturally associate this to an element of $\mathbf{V}_{\text{curv}}$. This defines the value of $\proj_{\mathbf{V}_{\text{curv}}}\circ R$. For $\proj_{T^{*}_{p}\slashed{\mathcal{N}}}\circ R$, we simply define this to be $f_{x}dx + f_{y}dy$.

Note that it is not at all clear that $R$ is affine linear. However, the computations done in Section \ref{3to2} imply that $\hat F \circ R$ is the same as the inclusion of $\widetilde{\mathbf{V}}_{\text{conf}}$ into ${\mathbf{V}}_{\text{conf}}$. More concretely, if we pick some function $f$, then if we compute the associated Ricci curvature (and derivatives), then we have seen that it is possible to recover the terms in the $3$-jet of $f$ at $p$ which are contained in $\mathbf{V}_{\text{curv}}$. Because we have seen that $F$ is a bijection, we thus find that $\proj_{\mathbf{V}_{\text{curv}}} \circ R$ is affine linear.

Finally, we define $\Xi : \mathbf{V}_{\text{curv}}\oplus T^{*}_{p}\slashed{\mathcal{N}} \to  T^{*}_{p}\slashed{\mathcal{N}} \oplus  T^{*}_{p}\slashed{\mathcal{N}}$, which maps the Ricci curvature and ``$df$'' to the Bianchi equations and $df-\alpha$. More precisely, we send
\[
\left(\sum_{a,b \in \{x,y\}} R_{ab}da db, \sum_{a,b,c\in\{x,y\}} \partial_{a}R_{bc}da \otimes (db\ dc), f_{x}dx+f_{y}dy\right)\mapsto (\mathfrak{B},(f_{x} - \alpha_{x})dx + (f_{y}-\alpha_{y})dy),
\]
where $\alpha$ and $\mathfrak{B}$ are formed by substituting in $R_{ab}$, $\partial_{a}R_{bc}$ and $f_{a}$ for the curvature, derivative of curvature, and differential of $f$ in the relevant equations. Direct inspection of the relevant formulas imply that $\Xi$ is surjective.

We are thus led to consider the following diagram of affine linear spaces\footnote{Note that the kernel of an affine map between vector spaces $a: \mathbf{V}_{1}\to\mathbf{V_{2}}$ can be defined to be the affine space $(a^{-1}(0), \mathbf{V}_{1})$. Here, we only consider the kernel of such maps. } and maps (where the horizontal sequences are exact\footnote{Recall that the horizontal arrows being \emph{exact} means that the composition of two horizontal maps is the zero map.}). We have defined the solid arrows above and we will define the dashed arrows below.
\[
\xymatrix{
0 \ar[r]& \ker \Xi \ar[r] \ar@/^/@{-->}[d]^{F'}& \mathbf{V}_{\text{curv}} \oplus T_{p}^{*}\slashed{\mathcal{N}}  \ar[r]^-{\Xi} \ar[d]^{\hat F} & T^{*}_{p}\slashed{\mathcal{N}}\oplus T^{*}_{p}\slashed{\mathcal{N}} \ar[r] \ar@{-->}[d]^{\beta} & 0 \\
0 \ar[r] & \widetilde{\mathbf{V}}_{\text{conf}}\ar[r] \ar[ur]_{R} \ar@/^/@{-->}[u]^{R'} & \mathbf{V}_{\text{conf}} \ar[r]^-{c} & \bigwedge^{2}T^{*}_{p}\slashed{\mathcal{N}} \oplus \mathbb{R} \ar[r] & 0
}
\]
By our above observations, the solid arrows commute. Moreover, by the computations in Section \ref{3to2}, we see that $\Xi \circ R = 0$. This is simply saying that if we start with an actual metric $-Q^{2}dU^{2} + e^{2f}(dx^{2}+dy^{2})$, then the Bianchi equations are satisfied, and $f$ satisfies the first order equation $df=\alpha$. Thus, we see that $R$ must lift to a map
\[
R' : \widetilde{\mathbf{V}}_{\text{conf}} \to \ker\Xi.
\]
Now, we claim that $\hat F$ lifts to a map $F' : \ker\Xi \to \widetilde{\mathbf{V}}_{\text{conf}}$ as indicated in the diagram. However, this follows the properties we have just proven, along with a straightforward dimension count:
\[
\dim\ker\Xi = (3+2\times 3 + 2) - (2+2) = 7 = (2+3+1+1+2) - (1+1) = \dim\widetilde{\mathbf{V}}_{\text{conf}}.
\]
In particular, we see that $F'$ is bijective. Because $\hat F\circ R$ equals the inclusion of $\widetilde{\mathbf{V}}_{\text{conf}}$ into ${\mathbf{V}}_{\text{conf}}$, we thus see that $R'$ is invertible, and $F'=(R')^{-1}$.

Now, the existence of an affine linear map $\beta: T^{*}_{p}\slashed{\mathcal{N}} \oplus T^{*}_{p}\slashed{\mathcal{N}}  \to \bigwedge^{2}T_{p}^{*}\slashed{\mathcal{N}} \oplus \mathbb{R}$ making the above diagram commute is clear: any element in $T^{*}_{p}\slashed{\mathcal{N}} \oplus T^{*}_{p}\slashed{\mathcal{N}} $ has a pre-image under $\Xi$. Pushing this element down by $\hat F$ and composing with $c$ yields an element of $\bigwedge^{2}T_{p}^{*}\slashed{\mathcal{N}} \oplus \mathbb{R}$. It is not hard to check that this yields $\beta$ as a well-defined,\footnote{Recall that the map being ``well defined'' in this context has to do with the possibility of choosing a different pre-image under $\Xi$. To show that the definition of the map is independent of this choice, note that the difference of two such elements lies in $\ker \Xi$, and is thus mapped to zero under $c\circ \hat F$. This shows that any other choice of pre-image would still yield the same element of $\wedge^{2}T^{*}_{p}\slashed{\mathcal{N}} \oplus \mathbb{R}$, so $\beta$ is ``well-defined.''} affine linear map. Moreover, because we have seen that $\hat F$ restricted to $\ker\Xi$ has image in $\widetilde{\mathbf{V}}_{\text{conf}}$, then it is easy to show that $\beta(0) = 0$. We thus define
\[
\proj_{\wedge^{2}T^{*}_{p}\slashed{\mathcal{N}}} \beta(0,dy) = (\beta_{2})_{x}dx\wedge dy, \qquad \proj_{\wedge^{2}T^{*}_{p}\slashed{\mathcal{N}}} \beta(0,dx) = -(\beta_{2})_{y}dx\wedge dy,
\]
and similarly for $\beta_{1}$. From this, the expression
\[
d\alpha = \beta_{1}\wedge \mathfrak{B} + \beta_{2} \wedge (df - \alpha)
\]
is simply a matter of un-winding the above diagram and definitions. Moreover, it is clear that $\beta_{1},\beta_{2}$ are uniquely defined by this prescription, and by inspecting the expression for $\alpha$, it is clear that their coefficients at $p$ are algebraic functions of the $3$-jet of $Q$ at $p$.

Finally, if $\mathfrak{B}=df-\alpha = 0$, then because $\alpha \in C^{1,\alpha}(\mathcal{U})$, we see that $f \in C^{2,\alpha}(\mathcal{U})$. Moreover, by considering the second factor in the compatibility map $c$, we see that the Laplacian defined as the trace of $J_{\alpha}$, the Jacobian of $\alpha$, and the Laplacian defined from the combination of $Q$ and $R_{ab}$ on the right hand side of the Lioiville equation must agree. Because $df=\alpha$, we see that
\[
\partial^{2}_{x}f + \partial^{2}_{y} f = \tr J_{\alpha}.
\]
Putting this together, the Liouville equation holds for $f$.
\end{proof}

As we see from the proof, it is not hard to compute $\beta_{1},\beta_{2}$ explicitly, by unwinding their definitions. In particular, it is not hard to derive the following expressions:

\begin{corollary}\label{determine-beta2-compat}
For $\mathcal{U}$ and $Q$ as in Proposition \ref{compat-condition}, the $1$-form $\beta_{2} = (\beta_{2})_{x}dx + (\beta_{2})_{y}dy$ satisfies
\begin{align*}
(\beta_{2})_{x} & = \frac{1}{2} \frac{(\partial_{y}Q)(\partial^{3}_{x,y,y}Q - \partial^{3}_{x}Q) +(\partial_{x}Q)(\partial^{3}_{y}Q - \partial^{3}_{x,x,y}Q + 2\partial^{2}_{x}Q + 2\partial^{2}_{y}Q)}{(\partial_{x}Q)^{2} + (\partial_{y}Q)^{2}} \\
(\beta_{2})_{y} & = \frac 12 \frac{(\partial_{x}Q)(\partial^{3}_{x,x,y}Q-\partial^{3}_{y}Q) + (\partial_{y}Q)(\partial^{3}_{x}Q - \partial^{3}_{x,y,y}Q + 2 \partial^{2}_{x}Q + 2\partial^{2}_{y}Q)}{(\partial_{x}Q)^{2}+(\partial_{y}Q)^{2}}
\end{align*}
\end{corollary}

\section{The Einstein Equations in $\mathcal{M}$}\label{theEquations}
Throughout this section, we let $(\mathcal{M},g)$ denote a stationary and axisymmetric spacetime whose metric, by definition, is of the form
\begin{equation}\label{metric}
g = -Vdt^2 + 2Wdtd\phi + Xd\phi^2 + e^{2\lambda}\left(d\rho^2 + dz^2\right).
\end{equation}
We set $\Phi \doteq \partial_{\phi}$ and $T \doteq \partial_t$.

Our goal is to use the results of Sections~\ref{geomprelim} and~\ref{curvcalc} in order to express the Einstein equations~(\ref{einsteinmatter}) as an explicit system of PDE's, from which Theorem \ref{theo:ein-reduce} will follow. 

Before we begin our reduction let's quickly orient ourselves with respect to the calculations of Section~\ref{curvcalc}.

Since $\Phi$ is non-vanishing in $\mathcal{M}$, we can apply the analysis of Section~\ref{4to3}. We obtain the corresponding submersion $(\overline{M},\overline{g})$ and then the Lorentzian manifold $(\hat{\mathcal{M}},\hat{g})$ by applying a conformal transformation $\hat{g} \doteq X\overline{g}$. Note that the metric $\hat{g}$ takes the form
\begin{equation}\label{hatg}
\hat{g} = -\left(XV+W^2\right)dt^2 + Xe^{2\lambda}\left(d\rho^2+dz^2\right).
\end{equation}

Let us also agree to the convention that $dt\wedge d\phi\wedge d\rho\wedge dz$ is positively oriented.

Our analysis will naturally lead to equations involving the Ricci tensor of $g$. The following well-known lemma allows us to express the Ricci curvature in terms of the energy momentum tensor.
\begin{lemma}\label{riccalc}Suppose that $(\mathcal{M},g)$ satisfies the Einstein equations~(\ref{einsteinmatter}). Then
\[Ric = \mathbb{T} - \frac{1}{2}g{\rm Tr}\left(\mathbb{T}\right).\]
\end{lemma}

Let us also take the opportunity to record the inverse of $g$ and the volume form.
\begin{lemma}\label{inverseg}We have
\[g^{-1} = -\sigma^{-2}X(T\otimes T) + \sigma^{-2}W (T\otimes \Phi) + \sigma^{-2}W(\Phi \otimes T) + \sigma^{-2}V(\Phi\otimes \Phi) + e^{-2\lambda}\left(\partial_{\rho}\otimes\partial_{\rho} + \partial_z\otimes\partial_z\right).\]
\end{lemma}
\begin{lemma}\label{volumeform}We have
\[dVol = \sigma e^{2\lambda} dt\wedge d\phi\wedge d\rho\wedge dz.\]
\end{lemma}
\subsection{The Equation for $X$}
In this section we will derive the equation for $X$.
\begin{proposition}We have
\begin{equation}\label{Xeqn1}
\sigma^{-1}\partial_{\rho}(\sigma\partial_{\rho}X) + \sigma^{-1}\partial_{z}(\sigma\partial_{z}X) = e^{2\lambda}(-2\mathbb{T}(\Phi,\Phi)+{\rm Tr}\left(\mathbb{T}\right)X) + \frac{(\partial_{\rho}X)^{2}+(\partial_{z}X)^{2} -\theta_{\rho}^{2}-\theta_{z}^{2}}{X}.
\end{equation}
\end{proposition}
\begin{proof}It follows immediately from Lemma~\ref{bochner} that
\begin{align*}
\Box X &= -2Ric\left(\Phi,\Phi\right) + \frac{\left[\left|\nabla X\right|^2 - \left|\theta\right|^2\right]}{X}
\\ \nonumber &=-2\mathbb{T}\left(\Phi,\Phi\right) + {\rm Tr}\left(\mathbb{T}\right)X + e^{-2\lambda}\frac{(\partial_{\rho}X)^2 + (\partial_zX)^2 - \theta_\rho^2 - \theta_z^2}{X}.
\end{align*}
Next, a straightforward coordinate computation yields
\[\Box X = e^{-2\lambda}\sigma^{-1}\partial_{\rho}\left(\sigma\partial_{\rho}X\right) + e^{-2\lambda}\sigma^{-1}\partial_z\left(\sigma\partial_zX\right).\]
Putting this together yields the desired expression.
\end{proof}
\subsection{The Equations for $\theta$ and $W$}
In this section we will derive the equations for $\theta$ and $W$.

We begin by giving an explicit formula for $\theta$:
\begin{proposition}We have
\[\theta = \sigma^{-1}X^2\left[\partial_{\rho}\left(X^{-1}W\right)dz - \partial_z\left(X^{-1}W\right)d\rho\right].\]
\end{proposition}
\begin{proof}It is straightforward to derive the following formula:
\[\nabla\Phi_{\flat} = -\frac{1}{2}\partial_{\rho}Wdt\wedge d\rho - \frac{1}{2}\partial_zWdt\wedge dz - \frac{1}{2}\partial_{\rho}Xd\phi\wedge d\rho - \frac{1}{2}\partial_zXd\phi\wedge dz.\]
Now, noting that
\[dt^{\sharp} = -\sigma^{-2}XT + \sigma^{-2}W\Phi,\]
\[d\phi^{\sharp} = \sigma^{-2}WT + \sigma^{-2}V\Phi,\]
one computes
\[*\left(dt\wedge d\rho\right) = \sigma^{-1}Xd\phi\wedge dz +\sigma^{-1}W dt\wedge dz,\]
\[*\left(dt\wedge dz\right) = -\sigma^{-1}Xd\phi\wedge d\rho - \sigma^{-1}Wdt\wedge d\rho,\]
\[*\left(d\phi \wedge d\rho\right) = -\sigma^{-1}Wd\phi\wedge dz +\sigma^{-1}Vdt\wedge dz,\]
\[*\left(d\phi \wedge dz\right) = \sigma^{-1}Wd\phi\wedge d\rho -\sigma^{-1}Vdt\wedge d\rho.\]
Then the proof concludes from a straightforward computation using $\theta = 2i_{\Phi}\left(*\nabla\Phi_{\flat}\right)$.
\end{proof}

An immediate corollary of the above is an equation for $W$ in terms of $\theta$.
\begin{corollary}We have
\begin{equation}\label{weqn}
d\left(X^{-1}W\right) = \frac{\sigma}{X^2}\left[\theta_{\rho}dz - \theta_zd\rho\right].
\end{equation}
\end{corollary}

Next we compute $d\theta$.
\begin{proposition}We have
\begin{equation}\label{dofthetheta}
d\theta = (\partial_{\rho}\theta_{z} - \partial_{z}\theta_{\rho}) d\rho\wedge dz = 2\sigma^{-1}e^{2\lambda}\left(\mathbb{T}\left(\Phi,\Phi\right)W - \mathbb{T}\left(\Phi,T\right)X\right)d\rho\wedge dz.
\end{equation}
\end{proposition}
\begin{proof}
Lemma~\ref{dtheta} implies
\[d\theta = 2i_{\Phi}*i_{\Phi}Ric.\]
Note that the analysis of Section~\ref{curvcalc} implies that
\[Ric\left(\Phi,\partial_{\rho}\right) = 0,\]
\[Ric\left(\Phi,\partial_z\right) = 0.\]
A straightforward calculation thus yields
\begin{align*}
2i_{\Phi}*i_{\Phi}Ric &= 2\sigma^{-1}e^{2\lambda}\left(Ric(\Phi,\Phi)W-Ric\left(\Phi,T\right)X\right)d\rho\wedge dz
\\ \nonumber &= 2\sigma^{-1}e^{2\lambda}\left(\mathbb{T}\left(\Phi,\Phi\right)W - \mathbb{T}\left(\Phi,T\right)X\right)d\rho\wedge dz.\qedhere
\end{align*}
\end{proof}
In order to completely determine $\theta$ we also need to compute the divergence.
\begin{proposition}

We have
\begin{equation}\label{divtheta}
\sigma^{-1}\partial_{\rho}\left(\sigma\theta_{\rho}\right) + \sigma^{-1}\partial_z\left(\sigma\theta_z\right) = \frac{2\theta_{\rho}\partial_{\rho}X + 2\theta_z\partial_zX}{X}.
\end{equation}
\end{proposition}
\begin{proof}Lemma~\ref{divdiv} yields
\begin{equation*}
{\rm div}\left(\theta\right) = 2\frac{g\left(\theta,\nabla X\right)}{X}.
\end{equation*}
An easy calculation then yields the formula~(\ref{divtheta}).
\end{proof}
\begin{remark}Note that~(\ref{divtheta}) is exactly the necessary compatibility condition to solve~(\ref{weqn}) and that the energy-momentum tensor $\mathbb{T}$ does not appear in the equation.
\end{remark}

\subsection{The Equation for $\sigma$}
In this section we will determine the equation for $\sigma$.
\begin{proposition}
We have
\begin{equation}\label{sigmaeqn}
X^{-1}e^{-2\lambda}\sigma (\partial^{2}_{\rho}\sigma + \partial^{2}_{z}\sigma) = \mathbb{T}(T-X^{-1}W\Phi, T-X^{-1}W\Phi) - X^{-2}\sigma^{2}\mathbb{T}(\Phi,\Phi) + X^{-1}\sigma^{2}{\rm Tr}\left(\mathbb{T}\right).
\end{equation}
\end{proposition}
\begin{proof}It follows immediately from Corollary~\ref{deltQRic} that
\[\sigma\slashed{\Delta}\sigma = \hat{Ric}\left(T,T\right).\]
Then, keeping in mind that we have $\pi^{-1}\left(T\right) = T - X^{-1}W\Phi$ (this follows from an easy calculation and the definition of a submersion), Proposition~\ref{hatRic} then implies that
\begin{align*}
X^{-1}e^{-2\lambda}\sigma\left(\partial_{\rho}^2\sigma + \partial_z^2\sigma\right) &= Ric\left(T-X^{-1}W\Phi,T-X^{-1}W\Phi\right) - X^{-2}\sigma^2 Ric\left(\Phi,\Phi\right)
\\ \nonumber &= \mathbb{T}\left(T-X^{-1}W\Phi,T-X^{-1}W\Phi\right) - X^{-2}\sigma^2\mathbb{T}\left(\Phi,\Phi\right) + X^{-1}\sigma^2{\rm Tr}\left(\mathbb{T}\right).\qedhere
\end{align*}
\end{proof}
\subsection{The Equation for $\lambda$} In this section we will determine the equations for $\lambda$.
\begin{proposition}\label{lamstuff}
The conformal factor $\lambda$, satisfies the first order equations \[\partial_{\rho}\lambda = \alpha_{\rho} - \frac 12 \partial_{\rho}\log X, \qquad \partial_{z}\lambda = \alpha_{z} - \frac 12 \partial_{z} \log X,\]
where
\begin{align*}
& \left( (\partial_{\rho}\sigma)^{2}+(\partial_{z}\sigma)^{2}\right)  \alpha_{\rho} \\
& = \frac 12 (\partial_{\rho}\sigma) \sigma \left( \mathbb{T}(\partial_{\rho},\partial_{\rho}) - \mathbb{T}(\partial_{z},\partial_{z}) + \frac 12 X^{-2} \left[(\partial_{\rho}X)^{2} - (\partial_{z}X)^{2} + (\theta_{\rho})^{2} - (\theta_{z})^{2} \right] \right)\\
& + \partial_{\rho}\sigma (\partial^{2}_{\rho}\sigma - \partial^{2}_{z}\sigma) + \partial_{z}\sigma (\partial^{2}_{\rho,z}\sigma)\\
& + (\partial_{z}\sigma)\sigma \left[ \mathbb{T}(\partial_{\rho},\partial_{z}) + \frac 12 X^{-2}\left((\partial_{\rho}X)(\partial_{z}X) + (\theta_{\rho})(\theta_{z}) \right) \right],
\end{align*}
and
\begin{align*}
& \left( (\partial_{\rho}\sigma)^{2}+(\partial_{z}\sigma)^{2}\right) \alpha_{z} \\
& = -\frac 12 (\partial_{z}\sigma)\sigma \left( \mathbb{T}(\partial_{\rho},\partial_{\rho}) - \mathbb{T}(\partial_{z},\partial_{z}) + \frac 12 X^{-2} \left[(\partial_{\rho}X)^{2} - (\partial_{z}X)^{2} + (\theta_{\rho})^{2} - (\theta_{z})^{2} \right]  \right)\\
& - \partial_{z}\sigma (\partial^{2}_{\rho}\sigma - \partial^{2}_{z}\sigma) + \partial_{\rho}\sigma (\partial^{2}_{\rho,z}\sigma)\\
& + (\partial_{\rho}\sigma)\rho \left[ \mathbb{T}(\partial_{\rho},\partial_{z}) + \frac 12 X^{-2} \left( (\partial_{\rho}X)(\partial_{z}X) + (\theta_{\rho})(\theta_{z}) \right)\right].
\end{align*}
\end{proposition}
\begin{proof}
The Einstein equations imply that
\begin{align*}Ric\left(\partial_{\rho},\partial_{\rho}\right) & = \mathbb{T}(\partial_{\rho},\partial_{\rho}) - \frac 12e^{2\lambda} {\rm Tr}\left(\mathbb{T}\right) \\
Ric\left(\partial_{\rho},\partial_z\right) & = \mathbb{T}(\partial_{\rho},\partial_{z}) \\
Ric\left(\partial_z,\partial_z\right) & = \mathbb{T}(\partial_{z},\partial_{z}) - \frac 12e^{2\lambda} {\rm Tr}\left(\mathbb{T}\right)  \\
Ric\left(\Phi,\Phi\right) & = \mathbb{T}(\Phi,\Phi) - \frac 12 {\rm Tr}\left(\mathbb{T}\right) X.
\end{align*}
Using Proposition~\ref{hatRic}, we thus find
\begin{align*}
\hat{Ric}(\partial_{\rho},\partial_{\rho}) & = \mathbb{T}(\partial_{\rho},\partial_{\rho}) - \frac 12e^{2\lambda} {\rm Tr}\left(\mathbb{T}\right) + X^{-1}\left( \mathbb{T}(\Phi,\Phi) - \frac 12 {\rm Tr}\left(\mathbb{T} \right)X  \right)  e^{2\lambda}\\
& + \frac 12 X^{-2}\left[ (\partial_{\rho} X)^{2} + (\theta_{\rho})^{2}\right],\\
\hat{Ric}(\partial_{\rho},\partial_{z}) & = \mathbb{T}(\partial_{\rho},\partial_{z}) + \frac 12 X^{-2}\left[ (\partial_{\rho}X)(\partial_{z}X) + (\theta_{\rho})(\theta_{z}) \right],\\
\hat{Ric}(\partial_{z},\partial_{z}) & = \mathbb{T}(\partial_{z},\partial_{z}) - \frac 12e^{2\lambda} {\rm Tr}\left(\mathbb{T}\right) + X^{-1}\left(\mathbb{T}(\Phi,\Phi) - \frac 12 {\rm Tr}\left(\mathbb{T}\right) X  \right)e^{2\lambda}\\
& + \frac 12 X^{-2}\left[(\partial_{z}X)^{2} +(\theta_{z})^{2}\right].
\end{align*}
The proposition now follows from inserting these equations into Corollary \ref{first-order-eqns-for-f}.
\end{proof}

A similar argument, based on the Liouville equation in Lemma \ref{liouville-f-eqn}, yields
\begin{proposition}
The conformal factor $\lambda$ also satisfies the second order equation
\begin{align*}
2\partial^{2}_{\rho}\lambda + 2\partial^{2}_{z}\lambda & = - \partial^{2}_{\rho}\log X - \partial^{2}_{z}\log X +  \sigma^{-1}(\partial_{\rho}^{2}\sigma + \partial^{2}_{z}\sigma)\\
& + e^{2\lambda}{\rm Tr}(\mathbb{T}) - \mathbb{T}(\partial_{\rho},\partial_{\rho}) - \mathbb{T}(\partial_{z},\partial_{z})\\
& -2 X^{-1}\left(\mathbb{T}(\Phi,\Phi) - \frac 12 {\rm Tr}\left(\mathbb{T}\right) X\right)e^{2\lambda}\\
& - \frac 12 X^{-2}\left[ (\partial_{\rho}X)^{2} + (\partial_{z}X)^{2} + (\theta_{\rho})^{2} + (\theta_{z})^{2} \right].
\end{align*}
\end{proposition}

\subsection{The system implies the Einstein equations} Until now, we have showed that if the Einstein equations hold, then the metric data satisfies its relevant equations as described in Theorem \ref{theo:ein-reduce}. For the purpose of constructing solutions to the Einstein equations, it is actually the other direction which is most important. Assume that $|\partial\sigma|\not = 0$ and the metric data solves their respective equations. We claim that the associated metric $g$ is a solution to the Einstein equations.

We note that the metric ansatz/symmetry forces
\[
Ric(\Phi,\partial_{\rho})=Ric(\Phi,\partial_{z}) = Ric(T,\partial_{\rho}) = Ric(T,\partial_{z}) = 0.
\]
independently of the Einstein equations. The first two expressions follow from Lemma \ref{dtheta}, while the second two follow by combining Proposition \ref{hatRic} with Corollary \ref{hatRiemU}. Thus, combined with the metric ansatz and the assumption on the stress energy tensor in \eqref{eq:SEtensor-orthog-directions}, the Einstein equations are automatically satisfied for these pairs of vectors.

It is not hard to check that the other components of the Einstein equations are implied by the equations for the metric data. We only note that the intermediate fact that the equations imply that $\hat{Ric}(\partial_{\rho},\partial_{\rho}),\hat{Ric}(\partial_{\rho},\partial_{z}),\hat{Ric}(\partial_{z},\partial_{z})$ have the expected values (from which the the Einstein equations for $g$ follow easily) may be checked by observing that the Liouville equation and first order equation for the conformal factor are seen to be three (affine) linearly independent equations for these three components of $\hat{Ric}$.

\subsection{The reduced Bianchi equations}
Finally, to complete the proof of Theorem \ref{itallturnsouttobecompatible}, we claim that under the regularity conditions given there, if $X,W,\theta,\sigma$ satisfy their respective equations and $\mathbb{T}$ is divergence free with respect to $g$, then the Bianchi equations are satisfied for  in the sense that $\mathfrak{B} = 0$. To check this, after observing that the tensor $\mathbb{T}$ is horizontal, we may easily compute that the prescribed Ricci curvature equations for $(\hat{\mathcal{M}},\hat g)$ take the form
\[
\hat{Ric} - \frac 12 \hat R\hat g = \mathbb{T} + \frac 12 X^{-2}(dX\otimes dX + \theta\otimes \theta) - \frac 1 4 X^{-2}(|\hat \nabla X|^{2}_{\hat g} + |\theta|_{\hat g}^{2})\hat g.
\]
Now, a straightforward computation shows that the $\hat g$-divergence of the right hand side vanishes if $div_{g}(\mathbb{T}) = 0$. This shows that the Bianchi equations in the form given in Lemma \ref{bianchi-equations} are satisfied. Combined with the equation for $\sigma$, this implies that $\mathfrak{B}=0$. This completes the proof of Theorem \ref{itallturnsouttobecompatible}.

\appendix

\section{Coordinate Systems}\label{coords}
\subsection{The Axis}\label{axiscoor}
At the axis we use $(t,x,y,z) \in \mathbb{R} \times \mathbb{R}^3$ coordinates where
\[x \doteq \rho\cos\phi,\]
\[y \doteq \rho\sin\phi.\]
A computation yields that the metric becomes
\begin{align*}
&-Vdt^2 - 2\left(\frac{W}{\rho^2}\right)ydtdx + 2\left(\frac{W}{\rho^2}\right)xdtdy
\\ \nonumber &+\left(\frac{e^{2\lambda}x^2 + \left(\frac{X}{\rho^2}\right)y^2}{x^2+y^2}\right)dx^2 + \left(\frac{2\left(e^{2\lambda}-\left(\frac{X}{\rho^2}\right)\right)xy}{x^2+y^2}\right)dxdy + \left(\frac{e^{2\lambda}y^2 + \left(\frac{X}{\rho^2}\right)x^2}{x^2+y^2}\right)dy^2 + e^{2\lambda}dz^2.
\end{align*}
Some simplification yields
\begin{align*}
&-V_{\mathscr{A}}\left(x^2+y^2,z\right)dt^2 - 2W_{\mathscr{A}}\left(x^2+y^2,z\right)ydtdx + 2W_{\mathscr{A}}\left(x^2+y^2,z\right)xdtdy +
\\ \nonumber &\qquad \left(X_{\mathscr{A}}\left(x^2+y^2,z\right) + \Sigma_{\mathscr{A}}\left(x^2+y^2,z\right)x^2\right)dx^2 + 2\Sigma_{\mathscr{A}}\left(x^2+y^2,z\right)xydxdy +
\\ \nonumber &\qquad \left(X_{\mathscr{A}}\left(x^2+y^2,z\right) + \Sigma_{\mathscr{A}}\left(x^2+y^2,z\right)y^2\right)dy^2 + \left(X_{\mathscr{A}}\left(x^2+y^2,z\right) + \left(x^2+y^2\right)\Sigma_{\mathscr{A}}\left(x^2+y^2,z\right)\right)dz^2.
\end{align*}
This is manifestly regular.
\subsection{The Horizon}
First we change to dagger coordinates $(t^{\dagger},\phi^{\dagger},\rho,z)$ where
\[t^{\dagger} \doteq t,\qquad \phi^{\dagger} \doteq \phi - \Omega t.\]
The metric becomes
\[-V^{\dagger}\left(dt^{\dagger}\right)^2 + 2W^{\dagger}dt^{\dagger}d\phi^{\dagger} + X\left(d\phi^{\dagger}\right)^2 + e^{2\lambda}\left(d\rho^2+dz^2\right),\]
where
\[V^{\dagger} = V - 2\Omega W - \Omega^2X,\qquad W^{\dagger} = W + \Omega X.\]

Next we introduce Kruskal type coordinates $(w^{\plus},w^{\minus},\phi^{\dagger},z) \in [0,\infty) \times [0,\infty) \times (0,\pi) \times \mathbb{R}$ defined by
\[w^{\pm} \doteq \rho e^{\pm \kappa t^{\dagger}}.\]
A computation yields that the metric becomes
\[\frac{1}{4}\left(w^{\minus}\right)^2\Sigma_{\mathscr{H}}\left(w^{\plus}w^{\minus},z\right)\left(dw^{\plus}\right)^2 +\frac{1}{4}\left(w^{\plus}\right)^2\Sigma_{\mathscr{H}}\left(w^{\plus}w^{\minus},z\right)\left(dw^{\minus}\right)^2 + \] \[\left(\kappa^{-2}V_{\mathscr{H}}\left(w^{\plus}w^{\minus},z\right) + \frac{1}{2}w^{\plus}w^{\minus}\Sigma_{\mathscr{H}}\left(w^{\plus}w^{\minus},z\right)\right)dw^{\plus}dw^{\minus}\]
\[+ \kappa^{-1}W_{\mathscr{H}}\left(w^{\plus}w^{\minus},z\right)\left(w^{\minus}dw^{\plus} - w^{\plus}dw^{\minus}\right)d\phi^{\dagger} + X\left(d\phi^{\dagger}\right)^2 + e^{2\lambda}dz^2.\]
This is manifestly regular.
\subsection{Juncture of the Horizon with the Axis}
We only treat the extension near $p_N$ since the extension near $p_S$ is completely analogous. We begin by writing the metric in $(t^{\dagger},\phi^{\dagger},s,\chi)$ coordinates:
\[t^{\dagger} \doteq t,\qquad \phi^{\dagger} \doteq \phi - \Omega t,\]
\[\rho \doteq s\chi,\qquad z \doteq \frac{1}{2}\left(\chi^2-s^2\right) + \beta.\]
The metric becomes
\[-V^{\dagger}\left(dt^{\dagger}\right)^2 + 2W^{\dagger}dt^{\dagger}d\phi^{\dagger} + X\left(d\phi^{\dagger}\right)^2 + \left(\chi^2+s^2\right)e^{2\lambda}\left(ds^2+d\chi^2\right).\]
where
\[V^{\dagger} = V - 2\Omega W - \Omega^2X,\qquad W^{\dagger} = W + \Omega X.\]
Now we introduce a hybrid Kruskal--Euclidean coordinate system:
\[\hat{x} \doteq s\sin\phi^{\dagger},\qquad \hat{y} \doteq s\cos\phi^{\dagger},\qquad \hat{w}^{\pm} \doteq \chi e^{\pm\kappa t^{\dagger}}.\]
After some computations the metric becomes
\[\frac{1}{4}\left(w^{\minus}\right)^2\Sigma_N^{(2)}\left(x^2+y^2,w^{\plus}w^{\minus}\right)\left(dw^{\plus}\right)^2 + \frac{1}{4}\left(w^{\plus}\right)^2\Sigma_N^{(2)}\left(x^2+y^2,w^{\plus}w^{\minus}\right)\left(dw^{\minus}\right)^2 +\]
\[\left(\kappa^{-2}V_N\left(x^2+y^2,w^{\plus}w^{\minus}\right) + \frac{1}{2}w^{\plus}w^{\minus}\Sigma_N^{(2)}\left(x^2+y^2,w^{\plus}w^{\minus}\right)\right)dw^{\plus}dw^{\minus} + \]
\[\left(X_N\left(x^2+y^2,w^{\plus}w^{\minus}\right) + \Sigma_N^{(1)}\left(x^2+y^2,w^{\plus}w^{\minus}\right)x^2\right)dx^2 +\]
\[\left(X_N\left(x^2+y^2,w^{\plus}w^{\minus}\right) + \Sigma_N^{(1)}\left(x^2+y^2,w^{\plus}w^{\minus}\right)y^2\right)dy^2 + \]
\[\left(X_N\left(x^2+y^2,w^{\plus}w^{\minus}\right) + \Sigma_N^{(1)}\left(x^2+y^2,w^{\plus}w^{\minus}\right)\left(x^2+y^2\right)\right)dxdy -\]
\[\kappa^{-1}W_N\left(x^2+y^2,w^{\plus}w^{\minus}\right)\left(yw^{\minus}dw^{\plus}dx - xw^{\minus}dw^{\plus}dy - yw^{\plus}dw^{\minus}dx + xw^{\plus}dw^{\minus}dy\right).\]
This is manifestly regular.

\bibliography{bib}
\bibliographystyle{amsplain}

\end{document}